\documentclass[12pt,a4paper]{amsart}

\usepackage[american]{babel}
\usepackage{bm,amsmath,amssymb}
\usepackage{amsbsy,amsfonts,array}
\usepackage{verbatim,graphicx}
\usepackage{epstopdf}
\usepackage{subfigure}
\graphicspath{{./Figures/},{./}}
\usepackage{color}

\usepackage{textcomp}
\usepackage{wasysym}
\usepackage{pifont}
\usepackage{manfnt}
\usepackage{color}
\usepackage{colortbl}
\usepackage{mparhack}

\usepackage{url}
\usepackage{graphicx}
\usepackage{subfigure}
\usepackage{stmaryrd}
\usepackage{tikz}

\newtheorem{theorem}{Theorem}
\newtheorem{definition}[theorem]{Definition}
\newtheorem{lemma}[theorem]{Lemma}
\newtheorem{proposition}[theorem]{Proposition}
\newtheorem{remark}[theorem]{Remark}

\title[Solution Map for Organic Solar Cells]
{Solution Map Analysis of a Multiscale Drift-Diffusion Model for Organic Solar Cells}

\address{$^{2}$ Dipartimento di Matematica, Politecnico di Milano \\
				 Piazza L. da Vinci 32, 20133 Milano, Italy}
\author{Maurizio Verri$^{1}$ \and Matteo Porro$^{1}$ \and Riccardo Sacco$^{1}$ \and Sandro Salsa$^{1}$}
\email{maurizio.verri@polimi.it}
\email{matteo.porro1@polimi.it}
\email{riccardo.sacco@polimi.it}
\email{sandro.salsa@polimi.it}

\begin{document}

\date{\today}

\begin{abstract}
In this article we address the theoretical study of a multiscale
drift-diffusion (DD) model for the description of photoconversion mechanisms
in organic solar cells. The multiscale nature of the formulation is based on
the co-presence of light absorption, conversion and diffusion phenomena that
occur in the three-dimensional material bulk, of charge photoconversion
phenomena that occur at the two-dimensional material interface separating
acceptor and donor material phases, and of charge separation and subsequent
charge transport in each three-dimensional material phase to device
terminals that are driven by drift and diffusion electrical forces. The
model accounts for the nonlinear interaction among four species: excitons,
polarons, electrons and holes, and allows to quantitatively predict the
electrical current collected at the device contacts of the cell. Existence
and uniqueness of weak solutions of the DD system, as well as nonnegativity
of all species concentrations, are proved in the stationary regime via a
solution map that is a variant of the Gummel iteration commonly used in the
treatment of the DD model for inorganic semiconductors. The results are
established upon assuming suitable restrictions on the data and some
regularity property on the mixed boundary value problem for the Poisson
equation. The theoretical conclusions are numerically validated on the
simulation of three-dimensional problems characterized by realistic values
of the physical parameters.
\end{abstract}

\maketitle

{\bf Keywords:}
Organic semiconductors; solar cells; nonlinear systems
of partial differential equations; multi-domain formulation; Drift-Diffusion model;
functional iteration. 

\section{Introduction}\label{sec:introduction}

Within the widespread set of applications of nanotechnology, the branch of renewable energies
certainly occupies a prominent position because of the urgent need of addressing and solving
the problems related with the production and use of energy and its impact on air pollution
and climate. We refer to~\cite{IPCC_report}
for a realtime update of the state-of-the-art in the complex connection between industrial and
domestic usage of energy and global climate change.
Renewable energies comprise a set of different physical and technological approaches to production,
storage and delivery of sources of supply to everyday's life human activities that are alternative
to the usual fossile fuel, and include, without being limited to: solar, hydrogen, wind, biomass, geothermal
and tidal energies. A comprehensive survey on the fundamental role of nanotechnology in understanding
and developing novel advancing fronts in renewable energies can be found in~\cite{Hussein2015}.

In this article we focus our interest on the specific area of solar energy, and, more in detail,
on organic solar cells (OSCs). OSCs have received increasing attention in the current
nanotechnology industry because of distinguishing features, such as good efficiency at a very cheap cost
and mechanical flexibility because of roll-to-roll fabrication process, which make them promising
alternatives to traditional silicon-based devices~\cite{hoppe_sariciftci_2004}.
The macroscopic behaviour of an OSC depends strongly on the photoconversion mechanisms
that occur at much finer spatial and temporal scales, basically consisting in
(1) generation and diffusion of excited neutral states in the material bulk;
(2) dipole separation at material interfaces
into positive and negative charge carriers; and (3) transport of charge carriers in the
different material phases for subsequent collection of electric current at the output device terminals
(positive charges at the anode and negative charges at the cathode).
We refer to~\cite{deFalcoSacco2010,deFalco2012,porro2014} and references cited therein for a
physical description of the above mentioned phenomena, the mathematical analysis of some of their basic
functional properties and numerical implementation in a simulation tool.

In the following pages, we consider the model proposed and studied in~\cite{deFalco2012}, in
two-dimensional geometrical configurations, under the assumption that
the computational domain is a three-dimensional polyhedron divided into two disjoint regions
separated by a two-dimensional manifold that represents the material interface at which the
principal photoconversion phenomena take place. The structure considered in the present work
is described in Sect.~\ref{sec:geometry}
and can be regarded as a faithful representation of a realistic OSC.
The mathematical model, described in Sect.~\ref{sec:model_equations}, and then
subsequently in Sect.~\ref{sec:auxiliary_Poisson_pb} and Sect.~\ref{sec:stationary_multiscale_model},
is an extension of the classic Drift-Diffusion (DD) system of
partial differential equations (PDEs) used for the investigation of charge transport in semiconductor devices
for micro and nano-electronics~\cite{markowich1986stationary,Markowich,Jerome:AnalyCharTran,Lundstrom2009}.
It consists of a multidomain differential problem in conservation format
for four distinct species: excitons, polarons, electrons and holes. Excitons and polarons are neutral particles;
polarons may dissociate into electrons (negatively charged) and holes (positively charged) at the interface
and the resulting free charges are free to move in their respective material phases under the action of
a internal potential drop (related to the work function gap between the two phases) and of an external
electric field due to an applied voltage drop. Electrons and holes are
electrostatically coupled through Gauss' law in differential form (Poisson equation)
and kinetically coupled through recombination/generation reactions occurring at the interface.

The resulting problem is a highly
nonlinearly coupled system of advection-diffusion-reaction PDEs for which, in Sect.~\ref{sec:fixed_point_map},
we provide in the stationary regime a complete analysis of the existence and uniqueness of weak solutions, as well as nonnegativity of all species concentrations, via a
solution map that is a variant of the Gummel iteration commonly used in the
treatment of the DD model for inorganic semiconductors~\cite{Jerome:AnalyCharTran}.
The results are established upon assuming suitable restrictions on the data and some
regularity property on the mixed boundary value problem for the Poisson
equation. The theoretical conclusions are numerically validated in Sect.~\ref{sec:numerical_simulations}
on the simulation of three-dimensional problems characterized by realistic values
of the physical parameters whereas in Sect.~\ref{sec:conclusions}
some concluding remarks and indications for future extensions of model and analysis are illustrated.

\section{Geometry and notations}\label{sec:geometry}
Let $\Omega \subset \mathbb{R}^{3}$ denote the organic solar cell volume
(called from now on the \textit{device}). We assume that $\Omega $ is a bounded, connected, Lipschitzian open set.
\begin{figure}[h!]
\centering
\subfigure[]{
\includegraphics[width=0.45\textwidth,height=0.45\textwidth]{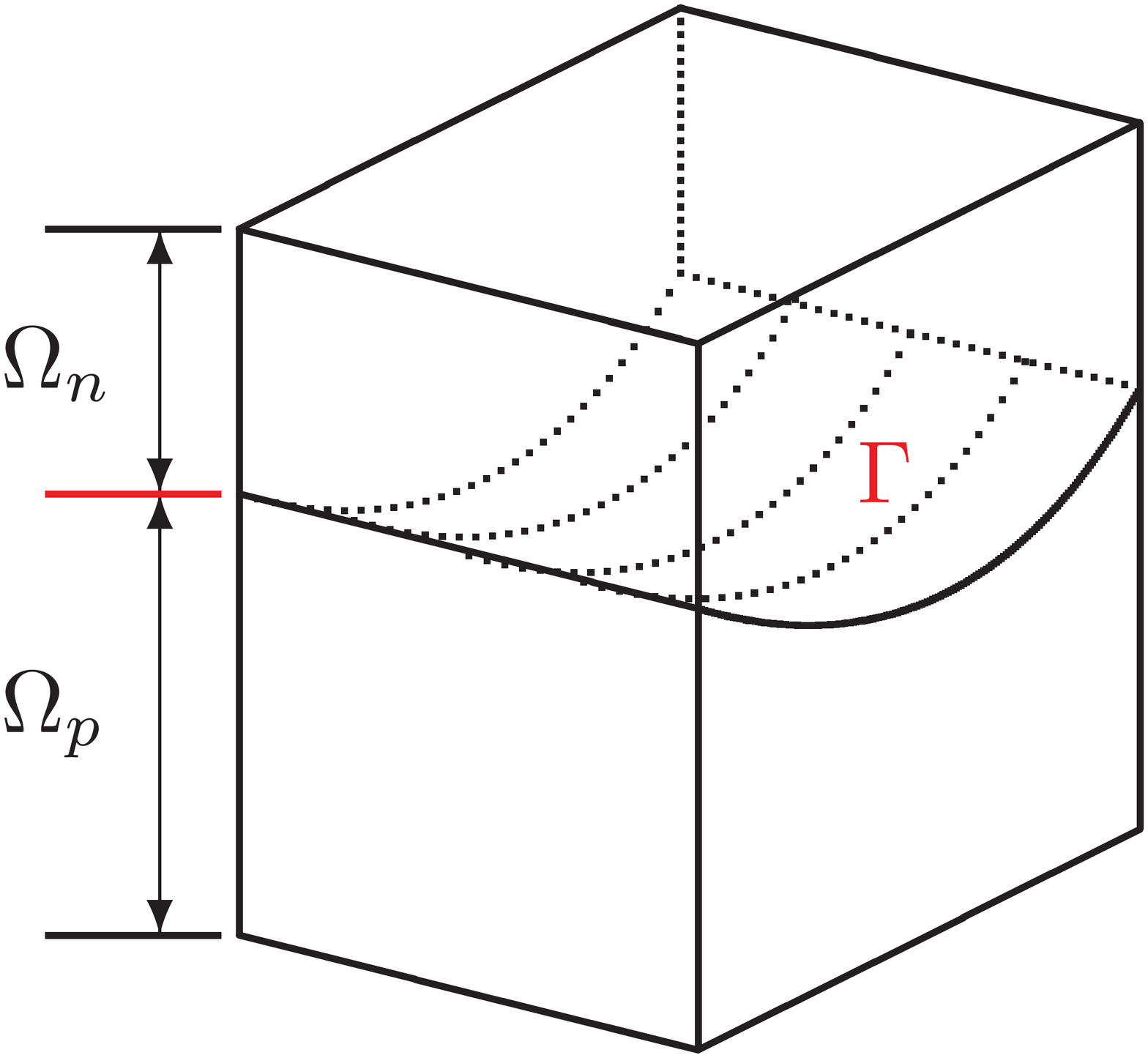}
\label{sfig:domain}}
\subfigure[]{
\includegraphics[width=0.50\textwidth,height=0.50\textwidth]{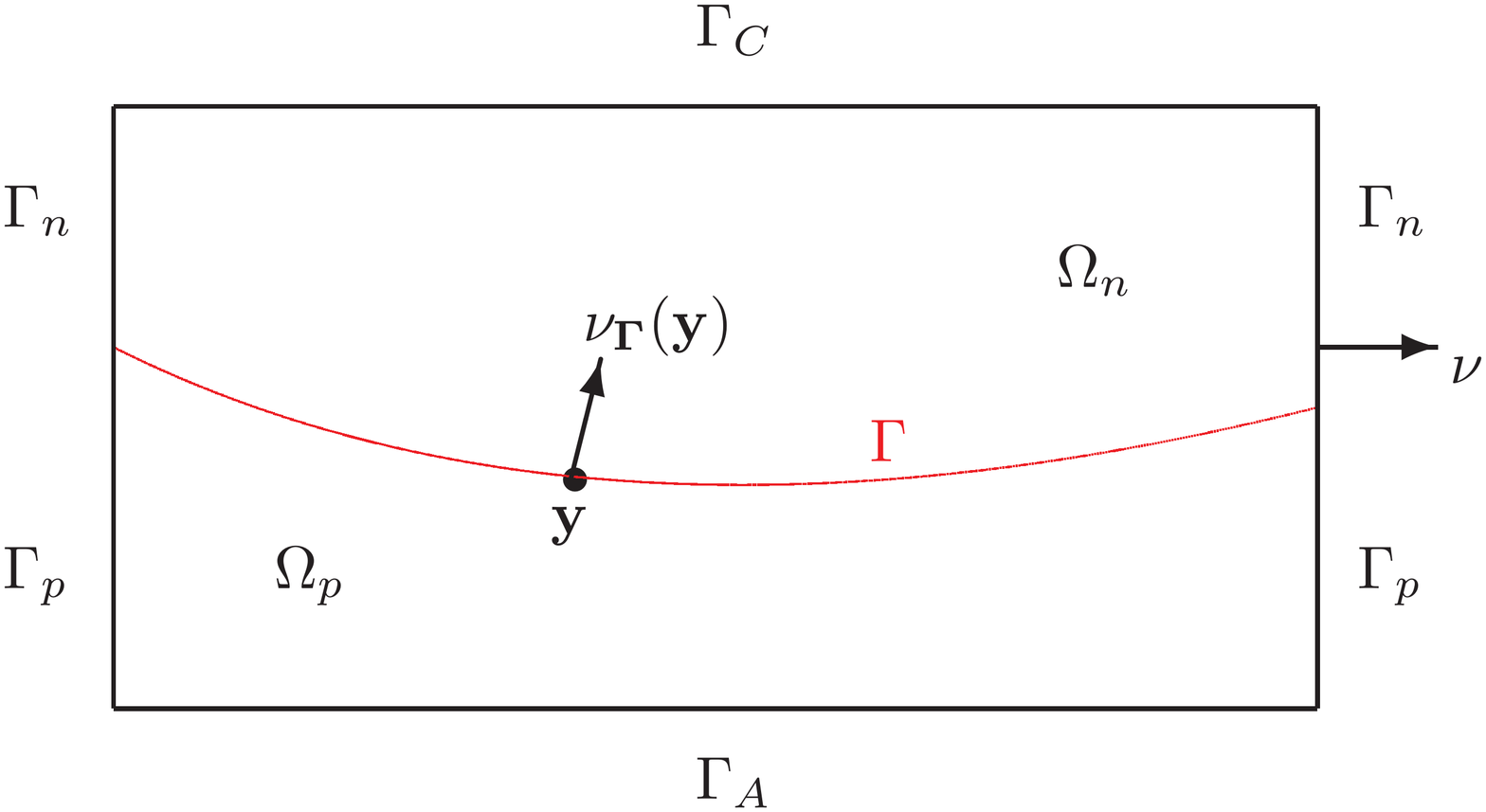}
\label{sfig:interface}}
\caption{Left: device domain. Right: domain boundary and interface.}
\label{fig:domain}
\end{figure}
Inside $\Omega$ we admit the presence of an open, regular surface $\Gamma$
(called from now on the \textit{interface}) that divides $\Omega$ into the two regions
(connected open sets) $\Omega_{n}$ and $\Omega_{p}$ in such a way that
$\Omega =\Omega_{n}\cup \Gamma \cup \Omega_{p}$. The unit normal vector
oriented from $\Omega_p$ into $\Omega_n$ is denoted by $\nu_\Gamma$.
A graphical plot of the three-dimensional (3D) domain comprising the interface is depicted in Fig.~\ref{sfig:domain}.
The boundary of $\Omega $ is the union of two disjoint subsets, so that
$\partial \Omega =\Gamma_{D}\cup \Gamma_{N}$. The
unit outward normal vector on $\partial \Omega$ is denoted by $\nu$.
Specifically, $\Gamma_{D}$ represents the contacts of the device, i.e.
\textit{anode} $\Gamma_{A}=\Gamma_{D}\cap \partial \Omega_{p}$ and
\textit{cathode} $\Gamma_{C}=\Gamma_{D}\cap \partial \Omega_{n}$. We assume that anode and cathode have nonzero
areas and that $\Gamma_{D}$ and $\Gamma$ are strictly separated. Furthermore, $\Gamma_{N}$ is the (relatively open) part of its boundary where the device is insulated from the surrounding environment. We put $\Gamma_{n}=\Gamma_{N}
\cap \partial \Omega_{n}$ and $\Gamma_{p}=\Gamma_{N}\cap \partial
\Omega_{p}$. A graphical plot of a two-dimensional (2D) cross-section of
the device domain comprising the interface and the boundary is depicted in Fig.~\ref{sfig:interface}.

The notation of function spaces in the present paper is as follows. We define $\mathcal{W}^{q}$ ($q\geq 2$) as the closure of the set%
\[
\left\{ \left. w\right\vert _{\Omega }:w\in C^{\infty }\left( \mathbb{R}%
^{3}\right) ,\text{ supp}\left( w\right) \cap \Gamma _{D}=\varnothing
\right\}
\]%
in $W^{1,q}\left( \Omega \right) $, that is, $\mathcal{W}^{q}$ is the
subspace of functions belonging to $W^{1,q}\left( \Omega \right) $ which
vanish on $\Gamma _{D}$ in the sense of traces%
\[
\mathcal{W}^{q}=\left\{ w\in W^{1,q}\left( \Omega \right) :\left.
w\right\vert _{\Gamma _{D}}=0\right\} .
\]%
Furthermore, we define $\mathcal{W}^{-q^{\prime }}\equiv \left( \mathcal{W}%
^{q}\right) ^{\prime }$ as the dual of $\mathcal{W}^{q}$ where $%
1/q+1/q^{\prime }=1$. $\mathcal{W}^{q}$ is a Banach space with respect to
the usual norm in $W^{1,q}\left( \Omega \right) $. Due to meas$\left( \Gamma
_{D}\right) >0$, the Poincar\'{e} inequality holds so that $\mathcal{W}^{q}$
can also be equipped with the equivalent norm%
\begin{equation}
\left\Vert w\right\Vert _{\mathcal{W}^{q}}=\left\Vert \nabla w\right\Vert
_{L^{q}\left( \Omega \right) }. \label{eq:norm_on_Wq}
\end{equation}%
In analogy with the definition of $\mathcal{W}^{q}$, we set%
\[
\mathcal{W}_{n}^{q}=\left\{ w\in W^{1,q}\left( \Omega _{n}\right) :\left.
w\right\vert _{\Gamma _{C}}=0\right\},
\]%
\[
\mathcal{W}_{p}^{q}=\left\{ w\in W^{1,q}\left( \Omega _{p}\right) :\left.
w\right\vert _{\Gamma _{A}}=0\right\}
\]%
with norms ($i=n,p$)%
\begin{equation}
\left\Vert w\right\Vert _{\mathcal{W}_{i}^{q}}=\left\Vert \nabla
w\right\Vert _{L^{q}\left( \Omega _{i}\right) }. \label{eq:norm_on_Wiq}
\end{equation}

\section{Model equations}\label{sec:model_equations}

In this section we illustrate the mathematical model of the OSC schematically represented
in Fig.~\ref{fig:domain}. For a detailed derivation of the equation system and the validation of its physical accuracy,
we invite the reader to consult~\cite{deFalco2012} and all references cited therein. For convenience, a list of all the variables and parameters of the cell model together with their units is contained in Tab.~\ref{tab:symbols}.
{\small
\begin{table}[h!]
\begin{tabular}{|l|l|l|}
\hline
\textsf{symbol} & \textsf{description} & \textsf{units} \\ \hline\hline
$e\left( t,\mathbf{x}\right) $ & $\text{concentration of excitons}$ & m$%
^{-3} $ \\ \hline
$n\left( t,\mathbf{x}\right) $ & $\text{concentration of electrons}$ & m$%
^{-3}$ \\ \hline
$p\left( t,\mathbf{x}\right) $ & $\text{concentration of holes}$ & m$^{-3}$
\\ \hline
$P\left( t,\mathbf{y}\right) $ & $\text{areal concentration
of polarons}$ & m$^{-2}$ \\ \hline
$\tau _{d}$ & exciton-polaron dissociation time & s \\ \hline
$\tau _{e}$ & exciton lifetime & s \\ \hline
$k_{d}$ & polaron dissociation rate & s$^{-1}$ \\ \hline
$k_{r}$ & polaron-exciton recombination rate & s$^{-1}$ \\ \hline
$\gamma $ & bimolecular recombination coefficient & m$^{3}$s$^{-1}$ \\ \hline
$\eta $ & polaron-exciton recombination fraction &  \\ \hline
$q$ & quantum of charge & C \\ \hline
$D_{e}$, $D_{n}$, $D_{p}$ & exciton (electron, hole) diffusion coefficient &
m$^{2}$s$^{-1}$ \\ \hline
$\mu _{n}$, $\mu _{p}$ & electron (hole) mobility & m$^{2}$V$^{-1}$s$^{-1}$
\\ \hline
$Q$ & exciton photogeneration rate & m$^{-3}$s$^{-1}$ \\ \hline
$\mathbf{J}_{e}=-D_{e}\nabla e$ & exciton flux density & m$^{-2}$s$^{-1}$ \\
\hline
$\mathbf{J}_{n}=q\left( D_{n}\nabla n+\mu _{n}n\mathbf{E}\right) $ &
electron current density & Cm$^{-2}$s$^{-1}$ \\ \hline
$\mathbf{J}_{p}=q\left( -D_{p}\nabla p+\mu _{p}p\mathbf{E}\right) $ & hole
current density & Cm$^{-2}$s$^{-1}$ \\ \hline
$\varphi \left( t,\mathbf{x}\right) $ & electric potential & V \\ \hline
$\mathbf{E}=-\nabla \varphi $ & electric field & Vm$^{-1}=$ NC$^{-1}$ \\
\hline
$E=\left\vert \mathbf{E} \right\vert =\left\vert \nabla \varphi \right\vert $ & electric field
intensity &  \\ \hline
$\widetilde{\varepsilon }$ & electric permittivity & CV$^{-1}$m$^{-1}=$ C$%
^{2}$N$^{-1}$m$^{-2}$ \\ \hline
$\varepsilon =\widetilde{\varepsilon }/q$ & electric permittivity per unit
charge & V$^{-1}$m$^{-1}=$ C$^{2}$N$^{-1}$m$^{-2}$ \\ \hline
$H$ & interface half- width & m \\ \hline
\end{tabular}%
\caption{Variables, coefficients and parameters of the solar cell model.}
\label{tab:symbols}
\end{table}
}

The equations for the description of exciton generation and dynamics inside the bulk of the device material read\footnote{%
We denote by $\llbracket f\rrbracket=\left. f\right\vert _{\Gamma \cap
\partial \Omega _{n}}-\left. f\right\vert _{\Gamma \cap \partial \Omega _{p}}
$ the jump of $f$ across $\Gamma $.}
\begin{subequations}\label{eq:excitons}
\begin{align}
& \dfrac{\partial e}{\partial t}-\nabla \cdot \left( D_{e}\nabla e\right) =
Q-\dfrac{e}{\tau _{e}} & \text{in }\Omega \setminus \Gamma \qquad
\text{for }t>0 \label{eq:cons_excitons} \\
& \llbracket e\rrbracket=0 & \text{on }\Gamma \qquad \text{for }t>0
\label{eq:jump_1_excitons} \\
& \llbracket-D_{e}\dfrac{\partial e}{\partial \nu _{\Gamma }}\rrbracket
=\eta k_{r}P-\dfrac{2H}{\tau _{d}}e & \text{on }\Gamma \qquad \text{for }t>0
\label{eq:jump_2_excitons} \\
& e=0 & \text{on }\Gamma_{D} \qquad \text{for }t>0
\label{eq:BC_D_excitons} \\
& \dfrac{\partial e}{\partial \nu }=0 &
\text{on }\Gamma_{N} \qquad \text{for }t>0
\label{eq:BC_N_excitons} \\
& e\left( 0,\mathbf{x}\right) =e_{0}
\left( \mathbf{x}\right) & \text{in }\Omega \qquad \text{for }t=0.
\label{eq:IC_excitons}
\end{align}
\end{subequations}

\begin{remark}
The boundary condition~\eqref{eq:BC_D_excitons} corresponds to assuming that perfect exciton quenching occurs at the contacts (see~\cite{walker2008}).
\end{remark}

The equations for the description of electron generation and dynamics inside the donor phase of the solar cell material read
\begin{subequations}\label{eq:electrons}
\begin{align}
& \dfrac{\partial n}{\partial t}-\nabla \cdot \left( D_{n}\nabla n-
\mu_{n}n\nabla \varphi \right) =0 &
\text{in }\Omega_{n} \qquad \text{for }t>0
\label{eq:cons_electrons} \\
& D_{n}\dfrac{\partial n}{\partial \nu _{\Gamma }}=\mu _{n}\dfrac{\partial
\varphi }{\partial \nu _{\Gamma }}n\mathbf{-}k_{d}P+2H\gamma np
& \text{on } \Gamma \qquad \text{for }t>0
\label{eq:gamma_electrons} \\
& n\equiv 0 & \text{in }\Omega_{p} \qquad \text{for }t>0
\label{eq:electrons_in_Omega_p} \\
&  n=0 & \text{on }\Gamma_{C} \qquad \text{for }t>0
\label{eq:BC_D_electrons} \\
&
D_{n}\dfrac{\partial n}{\partial \nu }=\mu _{n}\dfrac{\partial \varphi }{%
\partial \nu }n & \text{on }\Gamma_{n} \qquad
\text{for }t>0 \label{eq:BC_N_electrons} \\
&
n\left( 0,\mathbf{x}\right) =n_{0}\left( \mathbf{x}\right) & \text{in }%
\Omega_{n}\cup \Gamma \qquad \text{for }t=0.
\label{eq:IC_electrons}
\end{align}
\end{subequations}

\begin{remark}
The boundary condition~\eqref{eq:BC_D_electrons}
corresponds to assuming an infinite recombination velocity at the cathode.
\end{remark}

The equations for the description of hole generation and dynamics
inside the acceptor phase of the solar cell material read
\begin{subequations}\label{eq:holes}
\begin{align}
& \dfrac{\partial p}{\partial t}-\nabla \cdot \left( D_{p}\nabla p+
\mu_{p}p\nabla \varphi \right) =0  &
\text{in }\Omega_{p} \qquad \text{for }t>0
\label{eq:cons_holes} \\
& D_{p}\dfrac{\partial p}{\partial \nu _{\Gamma }}=-\mu _{p}\dfrac{\partial
\varphi }{\partial \nu _{\Gamma }}p\mathbf{+}k_{d}P-2H\gamma np
& \text{on } \Gamma \qquad \text{for }t>0
\label{eq:gamma_holes} \\
& p\equiv 0 & \text{in }\Omega_{n} \qquad \text{for }t>0
\label{eq:holes_in_Omega_n} \\
&  p=0 & \text{on }\Gamma_{A} \qquad \text{for }t>0
\label{eq:BC_D_holes} \\
&
D_{p}\dfrac{\partial p}{\partial \nu }=-\mu _{p}\dfrac{\partial \varphi }{%
\partial \nu }p & \text{on }\Gamma_{p} \qquad
\text{for }t>0 \label{eq:BC_N_holes} \\
&
p\left( 0,\mathbf{x}\right) =p_{0}\left( \mathbf{x}\right) & \text{in }%
\Omega_{p}\cup \Gamma \qquad \text{for }t=0.
\label{eq:IC_holes}
\end{align}
\end{subequations}

\begin{remark}
The boundary condition~\eqref{eq:BC_D_holes}
corresponds to assuming an infinite recombination velocity at the anode.
\end{remark}

The equations for the description of polaron generation and dynamics
on the interface separating the two material phases
of the solar cell material read
\begin{subequations}\label{eq:polarons}
\begin{align}
& P\equiv 0 & \text{in }\Omega_{n}\cup \Omega_{p} \qquad \text{for }t>0
& \label{eq:bulk_polarons} \\
& \dfrac{\partial P}{\partial t}=\dfrac{2H}{\tau _{d}}e+2H\gamma np-\left(
k_{d}+k_{r}\right) P
& \text{on } \Gamma \qquad \text{for }t>0
\label{eq:gamma_polarons} \\
&
P\left( 0,\mathbf{x}\right) =P_{0}\left( \mathbf{x}\right) & \text{on }
\Gamma \qquad \text{for }t=0.
\label{eq:IC_polarons}
\end{align}
\end{subequations}

The equations for the description of electric potential distribution
inside the bulk of the device material read
\begin{subequations}\label{eq:potential}
\begin{align}
& -\nabla \cdot \left( \varepsilon \nabla \varphi \right) =-n
& \text{in }\Omega_n \label{eq:pot_Omega_n} \\
& -\nabla \cdot \left( \varepsilon \nabla \varphi \right) =+p
& \text{in }\Omega_p \label{eq:pot_Omega_p} \\
& \llbracket \varphi \rrbracket=0 & \text{on }\Gamma
\label{eq:jump_1_pot} \\
& \llbracket\varepsilon \dfrac{\partial \varphi }{%
\partial \nu _{\Gamma }}\rrbracket=0 & \text{on }\Gamma
\label{eq:jump_2_pot} \\
& \varphi =\varphi_{C}\left( \mathbf{x}\right) & \text{on }\Gamma_{C}
\label{eq:BC_D_pot_C} \\
& \varphi =\varphi_{A}\left( \mathbf{x}\right) & \text{on }\Gamma_{A}
\label{eq:BC_D_pot_A} \\
& \dfrac{\partial \varphi }{\partial \nu }=0 &
\text{on }\Gamma_{n}\cup\Gamma_{p}.
\label{eq:BC_N_pot}
\end{align}
\end{subequations}

\begin{remark}
Condition~\eqref{eq:jump_1_pot} expresses the physical fact that
the potential is continuous passing from the acceptor to the donor material
phase of the cell. Condition~\eqref{eq:jump_2_pot}
means no charge density on the interface $\Gamma$.
\end{remark}

The general assumptions satisfied by all model coefficients and parameters throughout the paper are collected in Tab.~\ref{tab:assumptions}.

\begin{table}[h!]
\begin{tabular}{|l|l|l|}
\hline
\textsf{symbol} & \textsf{assumption} & \textsf{bounds} \\ \hline\hline
$\tau _{d}$ & constant & $\tau _{d}>0$ \\ \hline
$\tau _{e}$ & constant & $\tau _{e}>0$ \\ \hline
$k_{d}\left( \mathbf{y}\right) $ & measurable &
$k_{d}\left( \mathbf{\cdot }\right) \geq 0\quad $a.e. on $\Gamma$ \\ \hline
$k_{r}$ & constant & $k_{r}>0$ \\ \hline
$\gamma \left( \mathbf{y}\right) $ &
$\in L^{\infty}\left( \Gamma \right) $ & $\exists \underline{\gamma },\overline{\gamma}
>0\qquad \underline{\gamma }\leq \gamma \left( \mathbf{\cdot }\right) \leq
\overline{\gamma }\quad $a.e. on $\Gamma $ \\ \hline
$\eta $ & constant & $0\leq \eta \leq 1$ \\ \hline
$D_{e}\left( \mathbf{x}\right) $ &
$\in L^{\infty }\left(\Omega \right) $ & $\exists \underline{d}_{e},\overline{d}_{e}>0\qquad
\underline{d}_{e}\leq D_{e}\left( \mathbf{\cdot }\right) \leq \overline{d}_{e}
\quad $a.e. in $\Omega $ \\ \hline
$D_{n}\left( \mathbf{x}\right) $ &
$\in L^{\infty }\left(\Omega_{n}\right) $ & $\exists \underline{d}_{n},\overline{d}_{n}>0\qquad
\underline{d}_{n}\leq D_{n}\left( \mathbf{\cdot }\right) \leq \overline{d}_{n}
\quad $a.e. in $\Omega_{n}$ \\ \hline
$D_{p}\left( \mathbf{x}\right) $ &
$\in L^{\infty }\left(\Omega_{p}\right) $ & $\exists \underline{d}_{p},\overline{d}_{p}>0\qquad
\underline{d}_{p}\leq D_{p}\left( \mathbf{\cdot }\right) \leq \overline{d}_{p}
\quad $a.e. in $\Omega_{p}$ \\ \hline
$\mu _{n}\left( \mathbf{x},E\right) $ &
$\in Car\left( \overline{\Omega_{n}}
\times \mathbb{R}\right) $ & $\exists \overline{\mu }_{n}>0\qquad 0\leq \mu
_{n}\left( \mathbf{\cdot },E\right) \leq \overline{\mu }_{n}\quad $a.e. in
$\overline{\Omega_{n}}$, $\forall E\geq 0 $ \\ \hline
$\mu _{p}\left( \mathbf{x},E\right) $ &
$\in Car\left( \overline{\Omega_{p}}
\times \mathbb{R}\right) $ & $\exists \overline{\mu }_{p}>0\qquad 0\leq \mu
_{p}\left( \mathbf{\cdot },E\right) \leq \overline{\mu }_{p}\quad $a.e. in
$\overline{\Omega_{p}}$, $\forall E\geq 0 $ \\ \hline
$Q\left( \mathbf{x}\right) $ & $\in L^{2}\left( \Omega
\right) $ & $Q\left( \mathbf{\cdot }\right) \geq 0\quad $a.e. in $\Omega $ \\ \hline
$\varepsilon \left( \mathbf{x}\right) $ &
$\in L^{\infty}\left( \Omega \right) $ & $\exists \underline{\varepsilon },
\overline{\varepsilon }>0\qquad \underline{\varepsilon }\leq
\varepsilon \left( \mathbf{\cdot }\right) \leq \overline{\varepsilon }\quad $a.e. in $\Omega $
\\ \hline
$H\left( \mathbf{y}\right) $ &
$\in L^{\infty }\left(
\Gamma \right) $ & $\exists \overline{h}>0\qquad 0\leq H\left( \mathbf{\cdot }\right)
\leq \overline{h}\quad $a.e. on $\Gamma $
\\ \hline
\end{tabular}
\caption{Assumptions on model coefficients and parameters.}
\label{tab:assumptions}
\end{table}

\section{The auxiliary Poisson problem}\label{sec:auxiliary_Poisson_pb}

The elliptic boundary value problem for the electric potential~\eqref{eq:potential}
can be written in more compact form as
\begin{subequations}\label{poisson}
\begin{align}
& -\nabla \cdot \left( \varepsilon \left( \cdot \right) \nabla \varphi \right)
=g\left( n,p\right)
& \text{in }\Omega \diagdown \Gamma \label{eq:poisson_eq} \\
&  \llbracket\varphi \rrbracket=\llbracket\varepsilon \left( \cdot \right)
\dfrac{\partial \varphi }{\partial \nu _{\Gamma }}\rrbracket=0 & \text{on } \Gamma
\label{eq:poisson_interface} \\
&  \varphi =\varphi _{D} & \text{on }\Gamma_{D} \label{eq:poisson_BC_D} \\
&  \dfrac{\partial \varphi }{\partial \nu }=0 & \text{on }\Gamma_{N} \label{eq:poisson_BC_N}
\end{align}
\end{subequations}
where
\begin{equation}\label{g}
g\left( n,p\right) :=\left\{
\begin{array}{ll}
-n & \qquad \text{in }\Omega_{n} \\
+p & \qquad \text{in }\Omega_{p}%
\end{array}%
\right.
\end{equation}%
and%
\begin{equation}\label{phi_D}
\varphi _{D}:=\left\{
\begin{array}{ll}
\varphi _{C} & \qquad \text{on }\Gamma_{C} \\
\varphi _{A} & \qquad \text{on }\Gamma_{A}.
\end{array}%
\right.
\end{equation}%
We assume that the electric permittivity $\varepsilon $ is as specified in
Tab.~\ref{tab:assumptions} and that there exists $\widetilde{\varphi }\in
H^{1}\left( \Omega \right) $ whose trace on $\partial \Omega $ is equal to
$\varphi _{D}$ on $\Gamma_{D}$. Next, for the moment
let $g\in L^{2}\left( \Omega \right)$ be a given function and consider the
following linear elliptic transmission problem with mixed boundary conditions
(from now on referred to as \emph{auxiliary Poisson problem}):
\begin{subequations}\label{eq:aux poisson}
\begin{align}
& -\nabla \cdot \left( \varepsilon \left( \cdot \right) \nabla \varphi \right)
=g \left( \cdot \right)
& \text{in }\Omega \diagdown \Gamma \label{eq:aux poisson_eq} \\
&  \llbracket\varphi \rrbracket=\llbracket\varepsilon \left( \cdot \right)
\dfrac{\partial \varphi }{\partial \nu _{\Gamma }}\rrbracket=0 & \text{on } \Gamma
\label{eq:aux poisson_interface} \\
&  \varphi =\varphi _{D} & \text{on }\Gamma_{D} \label{eq:aux poisson_BC_D} \\
&  \dfrac{\partial \varphi }{\partial \nu }=0 & \text{on }\Gamma_{N}. \label{eq:aux poisson_BC_N}
\end{align}
\end{subequations}
Let $u=\varphi -\widetilde{\varphi }$.
Then the auxiliary problem (\ref{eq:aux poisson}) is equivalent to
\begin{subequations}\label{eq:aux2 poisson}
\begin{align}
& -\nabla \cdot \left( \varepsilon \left( \cdot \right) \nabla u \right)
=g\left( \cdot \right) +\nabla \cdot \left( \varepsilon \left( \cdot \right)
\nabla \widetilde{\varphi }\right) & \qquad \text{in }\Omega  \label{eq:aux2 poisson_eq} \\
&  u=0 & \text{on }\Gamma_{D} \label{eq:aux2 poisson_BC_D} \\
&  \dfrac{\partial u}{\partial \nu }=-\dfrac{\partial \widetilde{\varphi }}{%
\partial \nu } & \text{on }\Gamma_{N}. \label{eq:aux2 poisson_BC_N}
\end{align}
\end{subequations}

\begin{definition}
$u\in \mathcal{W}^{2}$ is called a variational solution to the auxiliary
Poisson problem (\ref{eq:aux2 poisson}) if%
\begin{equation}
a\left( u,v\right) =L\left( v\right) \qquad \forall v\in \mathcal{W}^{2}
\label{Poisson_weak}
\end{equation}%
where
\begin{align*}
& a\left( u,v\right) =\int_{\Omega }\varepsilon \left( \cdot \right) \nabla
u\cdot \nabla vdx & \qquad u, v \in \mathcal{W}^{2}, \label{eq:a_uv} \\
& L\left( v\right) =\int_{\Omega }g\left( \cdot \right) vdx-\int_{\Omega
}\varepsilon \left( \cdot \right) \nabla \widetilde{\varphi }\cdot \nabla vdx
& \qquad v \in \mathcal{W}^{2}.
\end{align*}
\end{definition}

It is easily verified that $a$ and $L$ satisfy the hypotheses
of the Lax-Milgram Lemma. As a consequence, the following result can be proved.
\begin{lemma}[auxiliary Poisson problem, \#1]
Assume $\varepsilon \left( \cdot \right)$ as specified in
Tab.~\ref{tab:assumptions}, $g\in L^{2}\left( \Omega \right) $ and that there exists $%
\widetilde{\varphi }\in H^{1}\left( \Omega \right) $ whose trace on $%
\partial \Omega $ is equal to $\varphi _{D}$ on $\Gamma_{D}$. Then there is
a unique weak solution $\varphi $ to problem (\ref{eq:aux poisson}) in the function class
$\varphi -\widetilde{\varphi }=u\in \mathcal{W}^{2}$ and the following estimate holds
\begin{equation}
\left\Vert u\right\Vert _{\mathcal{W}^{2}}\leq \frac{c }{\underline{\varepsilon }}\left\Vert g\right\Vert _{L^{2}\left( \Omega
\right) }+\frac{\overline{\varepsilon }}{\underline{\varepsilon }}\left\Vert
\nabla \widetilde{\varphi }\right\Vert _{L^{2}\left( \Omega \right) }
\label{u estimate}
\end{equation}%
for some $c=c\left( \Omega \right) >0$.
\end{lemma}

In order to prove the existence of a weak solution to the DD system we need
a stronger solution to the auxiliary Poisson problem. To this end, consider
the elliptic operator $-\nabla \cdot \varepsilon \nabla :\mathcal{W}%
^{2}\rightarrow \mathcal{W}^{-2}$ defined by%
\[
\left\langle -\nabla \cdot \left( \varepsilon \left( \cdot \right) \nabla
u\right) ,v\right\rangle _{\mathcal{W}^{-2}}:=a\left( u,v\right) ,\qquad
u,v\in \mathcal{W}^{2}
\]%
and use the same notation $-\nabla \cdot \varepsilon \nabla $ for the
restriction of this operator to the spaces $\mathcal{W}^{q}$ ($q>2$). Then,
it is clear that it is a continuous operator from $\mathcal{W}^{q}$ into $%
\mathcal{W}^{-q}\equiv \left( \mathcal{W}^{q^{\prime }}\right) ^{\prime }$.
However, it would be desirable that $-\nabla \cdot \varepsilon \nabla :%
\mathcal{W}^{q}\rightarrow \mathcal{W}^{-q}$ provides a topological
isomorphism for some $q>2$, i.e. a one-to-one continuous mapping of $%
\mathcal{W}^{q}$ onto $\mathcal{W}^{-q}$ for which the inverse mapping is
also continuous. Since it is well known that this isomorphism property is
actually an assumption on $\Omega $, $\Gamma _{D}$ and $\Gamma _{N}$
(see~\cite{markowich1986stationary,BEI1996} and~\cite{groger1985,groger1989,Dauge1992,HALLER2008,Hieber2008,GLI}),
we shall call $q-$\textit{admissible} any triple $\left\{
\Omega ,\Gamma _{D},\Gamma _{N}\right\} $ such that the stated property
holds.

\begin{lemma}[auxiliary Poisson problem, \#2]\label{Lemma Poisson 2}
Assume that $\left\{ \Omega ,\Gamma _{D},\Gamma _{N}\right\} $ is a $q-$admissible triple for some
$q> 2$, $\varepsilon \left( \cdot \right) $ as specified in
Tab.~\ref{tab:assumptions}, $g\in L^{q}\left( \Omega \right) $ and
that there exists $\widetilde{\varphi }\in W^{1.q}\left( \Omega \right) $ whose trace on
$\partial \Omega $ is equal to $\varphi _{D}$ on $\Gamma_{D}$. Then there is
a unique solution $\varphi $ to problem (\ref{eq:aux poisson}) in the function
class $\varphi -\widetilde{\varphi }=u\in \mathcal{W}^{q}$ and the
following estimate holds
\begin{equation}
\left\Vert u\right\Vert _{\mathcal{W}^{q}}\leq c\left\{ \left\Vert
g\right\Vert _{L^{q}\left( \Omega \right) }+\left\Vert \nabla \widetilde{%
\varphi }\right\Vert _{L^{q}\left( \Omega \right) }\right\}
\label{u estimate1}
\end{equation}%
for some $c=c\left( q,\Omega ,\varepsilon \right) >0$.
\end{lemma}

\begin{proof}
Set $\psi =g\left( \cdot \right) +\nabla \cdot \left( \varepsilon \left(
\cdot \right) \nabla \widetilde{\varphi }\right) $.
Then $\psi \in \left( W^{1,q^{\prime }}\left( \Omega
\right) \right) ^{\prime }$, the dual of $W^{1,q^{\prime }}\left( \Omega
\right) $ (see~\cite{ZIE1989}, Th. 4.3.2, p.186). But the inclusion
$\mathcal{W}^{q^{\prime }}\subset W^{1,q^{\prime }}\left( \Omega \right)$
implies $\left( W^{1,q^{\prime }}\left( \Omega \right) \right) ^{\prime }\subset
\left( \mathcal{W}^{q^{\prime }}\right) ^{\prime }\equiv \mathcal{W}^{-q}$
so that the right-hand side of (\ref{eq:aux2 poisson_eq}) is an element of
$\mathcal{W}^{-q}$. Then by $q-$admissibility there is a unique solution
$u\in \mathcal{W}^{q}$ to problem (\ref{eq:aux2 poisson}) and
\begin{equation*}
\left\Vert u\right\Vert _{\mathcal{W}^{q}}\leq c\left( q,\Omega ,\varepsilon
\right) \left\Vert \psi \right\Vert _{\mathcal{W}^{-q}}.
\end{equation*}%
Now, $q> 2$ implies $q^{\prime }< 2$ so that $\mathcal{W}^{q}\subset
\mathcal{W}^{2}\subset \mathcal{W}^{q^{\prime }}$. Then, for $v\in \mathcal{W}^{2}$,
we have by H\"{o}lder's inequality
\begin{align*}
& \left\vert \left\langle \psi ,v\right\rangle_{\mathcal{W}^{-q}}\right\vert
=\left\vert \int_{\Omega }\left( gv-\varepsilon \nabla \widetilde{\varphi }%
\cdot \nabla v\right) dx\right\vert & \\
&\leq \left\Vert g\right\Vert _{L^{q}\left( \Omega \right) }\left\Vert
v\right\Vert _{L^{q\prime }\left( \Omega \right) }+\overline{\varepsilon }%
\left\Vert \nabla \widetilde{\varphi }\right\Vert _{L^{q}\left( \Omega
\right) }\left\Vert \nabla v\right\Vert _{L^{q\prime }\left( \Omega \right) } &
\\
&\leq c\left( q,\Omega \right) \left( \left\Vert g\right\Vert _{L^{q}\left(
\Omega \right) }+\overline{\varepsilon }\left\Vert \nabla \widetilde{\varphi
}\right\Vert _{L^{q}\left( \Omega \right) }\right) \left\Vert v\right\Vert _{%
\mathcal{W}^{q^{\prime }}}. &
\end{align*}%
By density the above estimate holds for all $v\in \mathcal{W}^{q^{\prime }}$
hence%
\begin{equation*}
\left\Vert \psi \right\Vert _{\mathcal{W}^{-q}}\leq c\left( q,\Omega \right)
\left( \left\Vert g\right\Vert _{L^{q}\left( \Omega \right) }+\overline{%
\varepsilon }\left\Vert \nabla \widetilde{\varphi }\right\Vert _{L^{q}\left(
\Omega \right) }\right)
\end{equation*}%
and~\eqref{u estimate1} follows.
\end{proof}
\begin{remark}
Lemma \ref{Lemma Poisson 2} guarantees that $\varphi \in W^{1,q}\left(
\Omega \right) $, hence (the restrictions) $\nabla \varphi \in L^{q}\left(
\Omega _{n}\right) $ and $\nabla \varphi \in L^{q}\left( \Omega _{p}\right) $%
. Moreover, using~\eqref{eq:norm_on_Wq} and~\eqref{u estimate1}, we get
\begin{align}
& \left\Vert \nabla \varphi \right\Vert _{L^{q}\left( \Omega _{n}\right)
}\leq \left\Vert \nabla u\right\Vert _{L^{q}\left( \Omega _{n}\right)
}+\left\Vert \nabla \widetilde{\varphi }\right\Vert _{L^{q}\left( \Omega
_{n}\right) }\leq \left\Vert u\right\Vert _{\mathcal{W}^{q}}+\left\Vert
\nabla \widetilde{\varphi }\right\Vert _{L^{q}\left( \Omega _{n}\right) }
\notag \\
& \leq c\left( q,\Omega ,\varepsilon \right) \left( \left\Vert g\right\Vert
_{L^{q}\left( \Omega \right) }+\left\Vert \nabla \widetilde{\varphi }%
\right\Vert _{L^{q}\left( \Omega \right) }\right) +\left\Vert \nabla
\widetilde{\varphi }\right\Vert _{L^{q}\left( \Omega _{n}\right) }  \notag \\
& \leq \left\{ c\left( q,\Omega ,\varepsilon \right) +1\right\} \left(
\left\Vert g\right\Vert _{L^{q}\left( \Omega \right) }+\left\Vert \nabla
\widetilde{\varphi }\right\Vert _{L^{q}\left( \Omega \right) }\right)
\label{grad fi estim}
\end{align}%
and a similar estimate holds true for $\left\Vert \nabla \varphi \right\Vert
_{L^{q}\left( \Omega _{p}\right) }$.
\end{remark}

\section{The multiscale model in the stationary case}\label{sec:stationary_multiscale_model}
In this section we examine the multiscale model of Sect.~\ref{sec:model_equations}
in stationary conditions. This corresponds to setting to zero all partial
derivatives with respect to the time variable $t$ and to assuming that all coefficients
and unknowns depend on the sole spatial variable $\mathbf{x}$.

\subsection{Polarons}

Eq.~\eqref{eq:gamma_polarons} has the explicit stationary solution for $%
\mathbf{y}\in \Gamma $%
\begin{equation}\label{P_red_stat}
P\left( \mathbf{y}\right) =\dfrac{2H\left( \mathbf{y}\right) }{\left(
k_{d}\left( \mathbf{y}\right) +k_{r}\right) \tau _{d}}\,e\left( \mathbf{y}%
\right) +\dfrac{2H\left( \mathbf{y}\right) \gamma \left( \mathbf{y}\right) }{%
k_{d}\left( \mathbf{y}\right) +k_{r}}\,n\left( \mathbf{y}\right) p\left(
\mathbf{y}\right)
\end{equation}%
and this expression has to be inserted into the condition on $\Gamma $ of
the stationary problems for excitons, electrons and holes. This is done in
the next sections.

\subsection{The auxiliary exciton problem}
Upon inserting (\ref{P_red_stat}) into~\eqref{eq:jump_2_excitons}
the stationary problem for the excitons reads
\begin{subequations}\label{e_red_stat}
\begin{align}
& -\nabla \cdot \left( D_{e}\left( \cdot \right) \nabla e\right) +\tau
_{e}^{-1}e=Q\left( \cdot \right) & \qquad \text{in }\Omega \setminus \Gamma \\
& \llbracket e\rrbracket=0 & \qquad \text{on }\Gamma \\
& \llbracket D_{e}\left( \cdot \right) \dfrac{\partial e}{\partial \nu
_{\Gamma }}\rrbracket=\alpha \left( \cdot \right) e-\beta \left( \cdot
\right) f\left( n,p\right) & \text{on }\Gamma \\
&  e=0 & \text{on }\Gamma_{D} \\
&  \dfrac{\partial e}{\partial \nu }=0 & \text{on }\Gamma_{N}%
\end{align}
\end{subequations}%
where we have set (for all $\mathbf{y}\in \Gamma $)
\begin{align}
& \alpha \left( \mathbf{y}\right) :=\dfrac{2H\left( \mathbf{y}\right) }{\tau
_{d}}\times \dfrac{k_{d}\left( \mathbf{y}\right) +\left( 1-\eta \right) k_{r}%
}{k_{d}\left( \mathbf{y}\right) +k_{r}}=\dfrac{2H\left( \mathbf{y}\right) }{%
\tau _{d}}-\dfrac{\beta \left( \mathbf{y}\right) }{\gamma \left( \mathbf{y}%
\right) \tau _{d}} \label{eq:alfa}\\
& \beta \left( \mathbf{y}\right) :=\dfrac{2\eta k_{r}\gamma \left( \mathbf{y}%
\right) H\left( \mathbf{y}\right) }{k_{d}\left( \mathbf{y}\right) +k_{r}} \label{eq:beta}\\
& f\left( n,p\right) :=np. \label{eq:f=np}
\end{align}
Taking into account the bounds
stated in Tab.~\ref{tab:assumptions} the functions $\alpha $ and $\beta $
satisfy the following constraints:
\begin{align*}
& 0\leq \left( 1-\eta \right) \dfrac{2H\left( \cdot \right) }{\tau _{d}}\leq
\alpha \left( \cdot \right) \leq \dfrac{2H\left( \cdot \right) }{\tau _{d}}%
\leq \dfrac{2\overline{h}}{\tau _{d}}=:\overline{\alpha }
\label{eq:alfa_bar} \\
& 0\leq \beta \left( \cdot \right) \leq 2\eta H\left( \cdot \right) \gamma
\left( \cdot \right) \leq 2\eta \overline{h}\overline{\gamma }=:\overline{%
\beta }
\end{align*}
For the moment let $f$ be a given function. Then the
transmission problem (\ref{e_red_stat}) is referred to as
the \emph{auxiliary exciton problem}:
\begin{subequations}\label{e_red_stat1}
\begin{align}
& -\nabla \cdot \left( D_{e}\left( \cdot \right) \nabla e\right) +\tau
_{e}^{-1}e=Q\left( \cdot \right) & \qquad \text{in }\Omega \setminus \Gamma \\
& \llbracket e\rrbracket=0 & \qquad \text{on }\Gamma \\
& \llbracket D_{e}\left( \cdot \right) \dfrac{\partial e}{\partial \nu
_{\Gamma }}\rrbracket=\alpha \left( \cdot \right) e-\beta \left( \cdot
\right) f\left( \cdot \right) & \text{on }\Gamma \\
&  e=0 & \text{on }\Gamma_{D} \\
&  \dfrac{\partial e}{\partial \nu }=0 & \text{on }\Gamma_{N}.%
\end{align}
\end{subequations}%

\begin{definition}
$e\in \mathcal{W}^{2}$ is called a variational solution to the auxiliary
exciton problem (\ref{e_red_stat1}) if%
\begin{equation}
b\left( e,v\right) =\ell \left( v\right) \qquad \forall v\in \mathcal{W}^{2}
\label{e_weak}
\end{equation}%
where
\begin{align*}
& b\left( u,v\right) =\int_{\Omega }D_{e}\left( \cdot \right) \nabla u\cdot
\nabla vdx+\tau _{e}^{-1}\int_{\Omega }uvdx+\int_{\Gamma }\alpha \left(
\cdot \right) uvd\sigma & \label{eq:b_exciton} \\
& \ell \left( v\right) =\int_{\Omega }Q\left( \cdot \right) vdx+\int_{\Gamma
}\beta \left( \cdot \right) f\left( \cdot \right) vd\sigma. &
\end{align*}
\end{definition}

\begin{lemma}[Auxiliary exciton problem]
\label{Lemma Excitons} Let $D_{e}$, $\tau _{e}$, $Q$ be as specified in
Tab.~\ref{tab:assumptions}; $f\in L^{2}\left( \Gamma \right) $; $\alpha
,\beta \in L^{\infty }\left( \Gamma \right) $ and $0\leq \alpha \leq
\overline{\alpha }$,\ $0\leq \beta \leq \overline{\beta }$ a.e. on $\Gamma $
for some constants $\overline{\alpha },\overline{\beta }>0$. Then there is a
unique variational solution $e$ to (\ref{e_red_stat1}). If in addition $f\geq 0$ a.e. on $\Gamma $, then the solution $%
e\geq 0$ a.e. in $\Omega $.
\end{lemma}

\begin{proof}
(\textit{Existence and uniqueness}) By the Sobolev Imbedding Theorem on
submanifolds we have
$H^{1}\left( \Omega \right) \hookrightarrow L^{q}\left( \Gamma \right)$ for all
$q\in \left[ 2,4\right]$.
Thus, there exists a constant $c=c\left( q,\Omega ,\Gamma \right) $ such that%
\begin{equation*}
\left\Vert v\right\Vert _{L^{q}\left( \Gamma \right) }\leq c\left\Vert
v\right\Vert _{H^{1}\left( \Omega \right) }\qquad \qquad \forall v\in
H^{1}\left( \Omega \right)
\end{equation*}%
hence we obtain, in particular,
\begin{equation}
\left\Vert v\right\Vert_{L^{q}\left( \Gamma \right) }\leq c\left\Vert v
\right\Vert_{\mathcal{W}^{2}} \qquad \forall v\in \mathcal{W}^{2}. \label{v estim}
\end{equation}
Then, using (\ref{v estim}) with $q=2$, we have
\begin{eqnarray*}
\left\vert b\left( u,v\right) \right\vert &\leq &\int_{\Omega }\left\vert
D_{e}\nabla u\cdot \nabla v\right\vert dx+\tau _{e}^{-1}\int_{\Omega
}\left\vert uv\right\vert dx+\int_{\Gamma }\left\vert \alpha uv\right\vert
d\sigma \\
&\leq &\overline{d}_{e}\left\Vert \nabla u\right\Vert _{L^{2}\left( \Omega
\right) }\left\Vert \nabla v\right\Vert _{L^{2}\left( \Omega \right) }+\tau
_{e}^{-1}\left\Vert u\right\Vert _{L^{2}\left( \Omega \right) }\left\Vert
v\right\Vert _{L^{2}\left( \Omega \right) }+\overline{\alpha }\left\Vert
u\right\Vert _{L^{2}\left( \Gamma \right) }\left\Vert v\right\Vert
_{L^{2}\left( \Gamma \right) } \\
&\leq & \left( \overline{d}_{e}+\tau _{e}^{-1}c\left( \Omega \right) +c\left(
\Omega ,\Gamma \right) \overline{\alpha }\right) \left\Vert u\right\Vert _{%
\mathcal{W}^{2}}\left\Vert v\right\Vert _{\mathcal{W}^{2}}\qquad \qquad
\qquad \forall u,v\in \mathcal{W}^{2}.
\end{eqnarray*}%
This shows that $b\left( u,v\right)$ is continuous on $\mathcal{W}^{2}\times \mathcal{W}^{2}$.
Furthermore, $b\left( u,v\right)$ is coercive on $\mathcal{W}^{2}\times \mathcal{W}^{2}$
because
\begin{align*}
& b\left( v,v\right) \geq \underline{d}_{e}\int_{\Omega }\left\vert \nabla v\right\vert ^{2}dx
=\underline{d}_{e}\left\Vert v\right\Vert _{\mathcal{W}^{2}}^{2} & \qquad \forall v\in \mathcal{W}^{2}
\end{align*}%
(recall that $\tau _{e}>0$ and $\alpha \geq 0$). Finally,
$\ell \left( v\right)$ is continuous on $\mathcal{W}^{2}$ because
\begin{eqnarray*}
\left\vert \ell \left( v\right) \right\vert &\leq &\left\Vert Q\right\Vert
_{L^{2}\left( \Omega \right) }\left\Vert v\right\Vert _{L^{2}\left( \Omega
\right) }+\overline{\beta }\left\Vert f\right\Vert _{L^{2}\left( \Gamma
\right) }\left\Vert v\right\Vert _{L^{2}\left( \Gamma \right) } \\
&\leq &\left( c\left( \Omega \right) \left\Vert Q\right\Vert _{L^{2}\left(
\Omega \right) }+c\left( \Omega ,\Gamma \right) \overline{\beta }\left\Vert
f\right\Vert _{L^{2}\left( \Gamma \right) }\right) \left\Vert v\right\Vert _{%
\mathcal{W}^{2}}\qquad \qquad \qquad \forall v\in \mathcal{W}^{2}.
\end{eqnarray*}%
Then the assertion follows by the Lax-Milgram Lemma.

\noindent
(\textit{Positivity}) Define $e^{+}=\max \left\{ e,0\right\} $ and $%
e^{-}=\max \left\{ -e,0\right\} $. Then $e^{+},e^{-}\geq 0$ and
$e=e^{+}-e^{-}$. Since $e^{-}\in \mathcal{W}^{2}$, we can choose $v=e^{-}$ in (\ref{e_weak})
to get
\begin{equation*}
b\left( e^{+}-e^{-},e^{-}\right) =\ell \left( e^{-}\right).
\end{equation*}%
But $\ell \left( e^{-}\right) \geq 0$ so that%
\begin{equation*}
0\leq b\left( e^{-},e^{-}\right) \leq b\left( e^{+},e^{-}\right).
\end{equation*}%
Let $\Omega_{+}=\left\{ e\geq 0\right\} $ and $\Omega_{-}=\left\{ e\leq
0\right\} $: then $\left. e^{+}\right\vert _{\Omega_{-}}=0$, $\left.
e^{-}\right\vert _{\Omega_{+}}=0$, hence $e^{+}e^{-}=0$ in $\Omega =\Omega
_{+}\cup \Omega_{-}$. As a consequence we have also $e^{+}e^{-}=0$ in $%
\Gamma $ and $\nabla e^{+}\cdot \nabla e^{-}=0$ in $\Omega $, so that $%
b\left( e^{+},e^{-}\right) =0$. In conclusion $b\left( e^{-},e^{-}\right) =0$%
, from which it follows $e^{-}=0$, i.e. $e=e^{+}\geq 0$ in $\Omega $.
\end{proof}

\begin{remark}
From (\ref{e_weak}) where $v=e$ is chosen, we see that%
\begin{equation*}
\underline{d}_{e}\left\Vert e\right\Vert _{\mathcal{W}^{2}}^{2}\leq b\left(
e,e\right) =\ell \left( e\right) \leq \left( c\left( \Omega \right)
\left\Vert Q\right\Vert _{L^{2}\left( \Omega \right) }+c\left( \Omega
,\Gamma \right) \overline{\beta }\left\Vert f\right\Vert _{L^{2}\left(
\Gamma \right) }\right) \left\Vert e\right\Vert _{\mathcal{W}^{2}}
\end{equation*}%
hence the variational solution $e$ of (\ref{e_red_stat1}) satisfies the
estimate%
\begin{equation}
\left\Vert e\right\Vert _{\mathcal{W}^{2}}\leq \frac{c\left( \Omega \right)
}{\underline{d}_{e}}\left\Vert Q\right\Vert _{L^{2}\left( \Omega \right) }+%
\frac{c\left( \Omega ,\Gamma \right) }{\underline{d}_{e}}\overline{\beta }%
\left\Vert f\right\Vert _{L^{2}\left( \Gamma \right) }  \label{e estimate}
\end{equation}%
for some constants $c\left( \Omega \right) >0$, $c\left( \Omega ,\Gamma
\right) >0$.
\end{remark}

\subsection{The auxiliary electron problem}\label{sec:aux_elec_pb}
Upon inserting (\ref{P_red_stat}) into~\eqref{eq:gamma_electrons}
the stationary problem for the electrons reads
\begin{subequations}
\label{n_red_stat}
\begin{align}
& -\nabla \cdot \left( D_{n}\left( \cdot \right) \nabla n-\mu _{n}\left(
\cdot ,\left\vert \nabla \varphi \right\vert \right) n\nabla \varphi \right)
=0 & & \qquad \text{in }\Omega _{n} \\
& D_{n}\left( \cdot \right) \dfrac{\partial n}{\partial \nu _{\Gamma }}=\mu
_{n}\left( \cdot ,\left\vert \nabla \varphi \right\vert \right) \dfrac{%
\partial \varphi }{\partial \nu _{\Gamma }}n+h\left( \cdot ,p\right) n%
\mathbf{-}h_{e}\left( \cdot ,e\right)  & & \qquad \text{on }\Gamma  \\
& n=0 & & \qquad \text{on }\Gamma _{C} \\
& D_{n}\left( \cdot \right) \dfrac{\partial n}{\partial \nu }=\mu _{n}\left(
\cdot ,\left\vert \nabla \varphi \right\vert \right) \dfrac{\partial \varphi
}{\partial \nu }n & & \qquad \text{on }\Gamma _{n}
\end{align}
\end{subequations}%
where ($\mathbf{y}\in \Gamma $)
\begin{align}
& \omega \left( \mathbf{y}\right) :=\frac{k_{d}\left( \mathbf{y}\right) }{%
k_{d}\left( \mathbf{y}\right) +k_{r}}\dfrac{2H\left( \mathbf{y}\right) }{%
\tau _{d}}  \\
& h\left( \mathbf{y},p\right) :=\frac{\beta \left( \mathbf{y}\right) }{\eta }%
p \label{eq:h_elec}\\
& h_{e}\left( \mathbf{y},e\right) :=\omega \left( \mathbf{y}\right) e\label{eq:h_e_elec}
\end{align}%
and where $\beta $ is defined in (\ref{eq:beta}). Note that we have
\begin{equation}
0\leq \omega \left( \cdot \right) \leq \frac{2H\left( \cdot \right) }{\tau
_{d}}\leq \dfrac{2\overline{h}}{\tau _{d}}=\overline{\alpha }.
\label{sup omega}
\end{equation}
Now assume that the function $\varphi $ in (\ref{n_red_stat}) is given by $%
\varphi =u+\widetilde{\varphi }$ where $u$ is the solution of the auxiliary
Poisson problem (\ref{eq:aux2 poisson}). In addition, suppose that $\mu _{n}$, $%
h $ and $h_{e}$ are given and known (with $\mu _{n}$ satisfying the bounds of
Tab.~\ref{tab:assumptions}). Then the transmission problem (\ref{n_red_stat})
is referred to as the \emph{auxiliary electron problem}:
\begin{subequations}
\label{n_red_stat1}
\begin{align}
& -\nabla \cdot \left( D_{n}\left( \cdot \right) \nabla n-\mu _{n}\left(
\cdot \right) n\nabla \varphi \right) =0 & & \qquad \text{in }\Omega _{n} \\
& D_{n}\left( \cdot \right) \dfrac{\partial n}{\partial \nu _{\Gamma }}=\mu
_{n}\left( \cdot \right) \dfrac{\partial \varphi }{\partial \nu _{\Gamma }}%
n+h\left( \cdot \right) n\mathbf{-}h_{e}\left( \cdot \right)  & & \qquad
\text{on }\Gamma  \\
& n=0 & & \qquad \text{on }\Gamma _{C} \\
& D_{n}\left( \cdot \right) \dfrac{\partial n}{\partial \nu }=\mu _{n}\left(
\cdot \right) \dfrac{\partial \varphi }{\partial \nu }n & & \qquad \text{on }%
\Gamma _{n}.
\end{align}
\end{subequations}%

\begin{definition}
$n\in \mathcal{W}_{n}^{2}$ is called a variational solution to the auxiliary
electron problem (\ref{n_red_stat1}) if%
\begin{equation}
a_{n}\left( n,v\right) =L_{n}\left( v\right) \qquad \forall v\in \mathcal{W}%
_{n}^{2}  \label{n_weak}
\end{equation}%
where
\begin{align*}
& a_{n}\left( n,v\right) =\int_{\Omega_{n}}D_{n}\left( \cdot \right) \nabla
n\cdot \nabla vdx-\int_{\Omega_{n}}\mu _{n}\left( \cdot \right) n\left(
\nabla v\cdot \nabla \varphi \right) dx+\int_{\Gamma }h\left( \cdot \right)
nvd\sigma & \\
& L_{n}\left( v\right) =\int_{\Gamma }h_{e}\left( \cdot \right) vd\sigma. &
\end{align*}
\end{definition}

\begin{lemma}[Auxiliary electron problem]\label{Lemma Electrons}
Assume that $\left\{ \Omega ,\Gamma _{D},\Gamma
_{N}\right\} $ is a $q-$admissible triple for some $q\geq 3$; let $\varphi $
be given by Lemma \ref{Lemma Poisson 2}; $D_{n}$, $\mu _{n}\in L^{\infty
}\left( \Omega _{n}\right) $ and satisfying the bounds of Tab.~\ref%
{tab:assumptions}; $h,h_{e}\in L^{2}\left( \Gamma \right) $; $h\geq 0$ a.e.
on $\Gamma $. Then there is a constant $\delta >0$ such that if $\left\Vert
\nabla \varphi \right\Vert _{L^{q}\left( \Omega _{n}\right) }<\delta $ then
problem (\ref{n_red_stat1}) has a unique variational solution $n$. If in
addition $h_{e}\geq 0$ a.e. on $\Gamma $, then the solution $n\geq 0$ a.e.
in $\Omega _{n}$.
\end{lemma}

\begin{proof}
(\textit{Existence and uniqueness}) Let us show that $a_{n}\left( u,v\right) $ is continuous on $\mathcal{W}%
_{n}^{2}\times \mathcal{W}_{n}^{2}$. We have%
\begin{align}
& \left\vert a_{n}\left( u,v\right) \right\vert \leq \overline{d}%
_{n}\left\Vert \nabla u\right\Vert _{L^{2}\left( \Omega _{n}\right)
}\left\Vert \nabla v\right\Vert _{L^{2}\left( \Omega _{n}\right) }  \notag
\label{an} \\
& +\overline{\mu }_{n}\int_{\Omega _{n}}\left\vert u\right\vert \,\left\vert
\nabla v\right\vert \,\left\vert \nabla \varphi \right\vert
\,dx+\int_{\Gamma }h\left\vert uv\right\vert d\sigma .
\end{align}%
By virtue of the H\"{o}lder's inequality for three functions the following
estimate holds
\begin{equation}\label{3holder}
\int_{\Omega _{n}}\left\vert u\right\vert \,\left\vert \nabla v\right\vert
\,\left\vert \nabla \varphi \right\vert \,dx\leq \left\Vert u\right\Vert
_{L^{r}\left( \Omega _{n}\right) }\left\Vert \nabla v\right\Vert
_{L^{2}\left( \Omega _{n}\right) }\left\Vert \nabla \varphi \right\Vert
_{L^{q}\left( \Omega _{n}\right) }
\end{equation}%
where $1/r+1/q=1/2$. The continuity of the embedding $H^{1}\left( \Omega
_{n}\right) \rightarrow L^{r}\left( \Omega _{n}\right) $ ($2\leq r\leq 6$)
yields%
\begin{equation}\label{emb_est}
\left\Vert u\right\Vert _{L^{r}\left( \Omega _{n}\right) }\leq c\left(
q,\Omega _{n}\right) \left\Vert u\right\Vert _{H^{1}\left( \Omega
_{n}\right) }\leq c\left( q,\Omega _{n}\right) \left\Vert u\right\Vert _{%
\mathcal{W}_{n}^{2}}
\end{equation}%
and $2\leq r\leq 6$ implies $q\geq 3$. Therefore%
\begin{equation}
\int_{\Omega _{n}}\left\vert u\right\vert \,\left\vert \nabla v\right\vert
\,\left\vert \nabla \varphi \right\vert \,dx\leq c\left( q,\Omega
_{n}\right) \left\Vert \nabla \varphi \right\Vert _{L^{q}\left( \Omega
_{n}\right) }\left\Vert u\right\Vert _{\mathcal{W}_{n}^{2}}\left\Vert
v\right\Vert _{\mathcal{W}_{n}^{2}}.  \label{est1}
\end{equation}%
In addition, using the generalized H\"{o}lder's inequality and the
continuity of trace and embedding $H^{1}\left( \Omega _{n}\right)
\longrightarrow H^{1/2}\left( \partial \Omega _{n}\right) \longrightarrow
L^{4}\left( \partial \Omega _{n}\right) $, gives%
\begin{eqnarray}
\int_{\Gamma }h\left\vert uv\right\vert d\sigma  &\leq &\left\Vert
h\right\Vert _{L^{2}\left( \Gamma \right) }\left\Vert uv\right\Vert
_{L^{2}\left( \Gamma \right) }
\leq \left\Vert h\right\Vert _{L^{2}\left( \Gamma \right) }\left\Vert
u\right\Vert _{L^{4}\left( \Gamma \right) }\left\Vert v\right\Vert
_{L^{4}\left( \Gamma \right) }  \notag \\
&\leq &\left\Vert h\right\Vert _{L^{2}\left( \Gamma \right) }\left\Vert
u\right\Vert _{L^{4}\left( \partial \Omega _{n}\right) }\left\Vert
v\right\Vert _{L^{4}\left( \partial \Omega _{n}\right) }  \notag \\
&\leq &c\left( \Omega _{n}\right) \left\Vert h\right\Vert _{L^{2}\left(
\Gamma \right) }\left\Vert u\right\Vert _{H^{1}\left( \Omega _{n}\right)
}\left\Vert v\right\Vert _{H^{1}\left( \Omega _{n}\right) }  \notag \\
&\leq &c\left( \Omega _{n}\right) \left\Vert h\right\Vert _{L^{2}\left(
\Gamma \right) }\left\Vert u\right\Vert _{\mathcal{W}_{n}^{2}}\left\Vert
v\right\Vert _{\mathcal{W}_{n}^{2}}.  \label{est2}
\end{eqnarray}%
Inserting (\ref{est1}) and (\ref{est2}) into (%
\ref{an}) yields%
\begin{equation*}
\left\vert a_{n}\left( u,v\right) \right\vert \leq \left\{ \overline{d}%
_{n}+c\left( q,\Omega _{n}\right) \overline{\mu }_{n}\left\Vert \nabla
\varphi \right\Vert _{L^{q}\left( \Omega _{n}\right) }+c\left( \Omega
_{n}\right) \left\Vert h\right\Vert _{L^{2}\left( \Gamma \right) }\right\}
\left\Vert u\right\Vert _{\mathcal{W}_{n}^{2}}\left\Vert v\right\Vert _{%
\mathcal{W}_{n}^{2}}
\end{equation*}%
which proves the continuity of $a_{n}\left( u,v\right) $. Concerning the
coercivity of $a_{n}\left( u,v\right) $, we have for $v\in \mathcal{W}%
_{n}^{2}$ (recall that $h\geq 0$)%
\begin{eqnarray*}
a_{n}\left( v,v\right)  &=&\int_{\Omega _{n}}D_{n}\left\vert \nabla
v\right\vert ^{2}dx-\int_{\Omega _{n}}\mu _{n}\left( \nabla \varphi \cdot
\nabla v\right) vdx+\int_{\Gamma }hv^{2}d\sigma  \\
&\geq &\underline{d}_{n}\int_{\Omega _{n}}\left\vert \nabla v\right\vert
^{2}dx-\int_{\Omega _{n}}\mu _{n}\left( \nabla \varphi \cdot \nabla v\right)
vdx \\
&\geq &\underline{d}_{n}\left\Vert \nabla v\right\Vert _{L^{2}\left( \Omega
_{n}\right) }^{2}-\overline{\mu }_{n}\int_{\Omega _{n}}\left\vert \nabla
\varphi \right\vert \,\left\vert \nabla v\right\vert \,\left\vert
v\right\vert \,dx \\
&&\text{by (\ref{est1})} \\
&\geq &\underline{d}_{n}\left\Vert v\right\Vert _{\mathcal{W}%
_{n}^{2}}^{2}-c\left( q,\Omega _{n}\right) \overline{\mu }_{n}\left\Vert
\nabla \varphi \right\Vert _{L^{q}\left( \Omega _{n}\right) }\left\Vert
v\right\Vert _{\mathcal{W}_{n}^{2}}^{2}
\end{eqnarray*}%
hence%
\begin{equation*}
a_{n}\left( v,v\right) \geq \Lambda _{n}\left\Vert v\right\Vert _{\mathcal{W}%
_{n}^{2}}^{2}\qquad \qquad \qquad \forall v\in \mathcal{W}_{n}^{2}
\end{equation*}%
where%
\begin{equation}
\Lambda _{n}:=\underline{d}_{n}-c\left( q,\Omega _{n}\right) \overline{\mu }%
_{n}\left\Vert \nabla \varphi \right\Vert _{L^{q}\left( \Omega _{n}\right) }.
\label{Lambda_n}
\end{equation}%
\newline
Using again the continuity of trace and embedding allows us to prove that $%
L_{n}\left( v\right) $ is continuous on $\mathcal{W}_{n}^{2}$:
\begin{eqnarray*}
\left\vert L_{n}\left( v\right) \right\vert  &\leq &\left\Vert
h_{e}\right\Vert _{L^{2}\left( \Gamma \right) }\left\Vert v\right\Vert
_{L^{2}\left( \Gamma \right) }\leq \left\Vert h_{e}\right\Vert _{L^{2}\left(
\Gamma \right) }\left\Vert v\right\Vert _{L^{2}\left( \partial \Omega
_{n}\right) } \\
&\leq &c\left( \Omega _{n}\right) \left\Vert h_{e}\right\Vert _{L^{2}\left(
\Gamma \right) }\left\Vert v\right\Vert _{\mathcal{W}_{n}^{2}}\qquad \qquad
\qquad \forall v\in \mathcal{W}_{n}^{2}.
\end{eqnarray*}%
Then we conclude that the existence of a unique solution follows by the
Lax-Milgram Lemma provided that $\Lambda _{n}>0$, i.e if
\begin{equation}
\left\Vert \nabla \varphi \right\Vert _{L^{q}\left( \Omega _{n}\right)
}<\delta :=\frac{\underline{d}_{n}}{\overline{\mu }_{n}}c\left( q,\Omega
_{n}\right) .  \label{eq:Lambda_n_gt_0}
\end{equation}

\noindent (\textit{Positivity}) Define $n^{+}=\max \left\{ n,0\right\} $ and
$n^{-}=\max \left\{ -n,0\right\} $. Then $n^{+},n^{-}\geq 0$ and $%
n=n^{+}-n^{-}$. Since $n^{-}\in \mathcal{W}_{n}^{2}$, we can choose $v=n^{-}$
in (\ref{n_weak}) to get%
\begin{equation*}
a_{n}\left( n^{+}-n^{-},n^{-}\right) =L_{n}\left( n^{-}\right) .
\end{equation*}%
But $L_{n}\left( n^{-}\right) \geq 0$ so that%
\begin{equation*}
a_{n}\left( n^{-},n^{-}\right) \leq a_{n}\left( n^{+},n^{-}\right) .
\end{equation*}%
Let $\Omega _{n}^{+}=\left\{ n\geq 0\right\} $ and $\Omega _{n}^{-}=\left\{
n\leq 0\right\} $: then $\Omega _{n}=\Omega _{n}^{+}\cup \Omega _{n}^{-}$
and $\left. n^{+}\right\vert _{\Omega _{n}^{-}}=\left. n^{-}\right\vert
_{\Omega _{n}^{+}}=0$, hence $a_{n}\left( n^{+},n^{-}\right) =0$. In
conclusion%
\begin{equation*}
0\leq \Lambda _{n}\left\Vert \nabla n^{-}\right\Vert _{L^{2}\left( \Omega
_{n}\right) }^{2}\leq a_{n}\left( n^{-},n^{-}\right) \leq 0
\end{equation*}%
from which it follows $\nabla n^{-}=0$ in $\Omega _{n}$. Then $n^{-}=0$ in $%
\Omega _{n}$ i.e. $n=n^{+}\geq 0$ in $\Omega _{n}$.
\end{proof}
\begin{remark}
From (\ref{n_weak}) where $v=n$ is chosen, we see that%
\begin{equation*}
\Lambda _{n}\left\Vert n\right\Vert _{\mathcal{W}_{n}^{2}}^{2}\leq
a_{n}\left( n,n\right) =L_{n}\left( n\right) \leq c\left( \Omega _{n}\right)
\left\Vert h_{e}\right\Vert _{L^{2}\left( \Gamma \right) }\left\Vert
n\right\Vert _{\mathcal{W}_{n}^{2}}
\end{equation*}%
hence the variational solution $n$ of (\ref{n_red_stat1}) satisfies the
estimate
\begin{equation}
\left\Vert n\right\Vert _{\mathcal{W}_{n}^{2}}\leq \frac{c\left( \Omega
_{n}\right) }{\Lambda _{n}}\left\Vert h_{e}\right\Vert _{L^{2}\left( \Gamma
\right) }  \label{n estimate}
\end{equation}%
for some $c=c\left( \Omega _{n}\right) >0$.
\end{remark}

\subsection{The auxiliary hole problem}

Upon inserting (\ref{P_red_stat}) into~\eqref{eq:gamma_holes} the stationary
problem for the holes reads
\begin{subequations}
\label{p_red_stat}
\begin{align}
& -\nabla \cdot \left( D_{p}\left( \cdot \right) \nabla p+\mu _{p}\left(
\cdot ,\left\vert \nabla \varphi \right\vert \right) p\nabla \varphi \right)
=0 & & \qquad \text{in }\Omega _{p} \\
& D_{p}\left( \cdot \right) \dfrac{\partial p}{\partial \nu _{\Gamma }}=-\mu
_{p}\left( \cdot ,\left\vert \nabla \varphi \right\vert \right) \dfrac{%
\partial \varphi }{\partial \nu _{\Gamma }}p\mathbf{+}h_{e}\left( \cdot
,e\right) -h\left( \cdot ,n\right) p & & \qquad \text{on }\Gamma \\
& p=0 & & \qquad \text{on }\Gamma _{A} \\
& D_{p}\left( \cdot \right) \dfrac{\partial p}{\partial \nu }=-\mu
_{p}\left( \cdot ,\left\vert \nabla \varphi \right\vert \right) \dfrac{%
\partial \varphi }{\partial \nu }p & & \qquad \text{on }\Gamma _{p}
\end{align}%
where $h$ and $h_{e}$ are defined as in~\eqref{eq:h_elec} and~%
\eqref{eq:h_e_elec}. In analogy with the case of electrons, we consider the
\textit{auxiliary hole problem:}
\end{subequations}
\begin{subequations}
\label{p_red_stat1}
\begin{align}
& -\nabla \cdot \left( D_{p}\left( \cdot \right) \nabla p+\mu _{p}\left(
\cdot \right) p\nabla \varphi \right) =0 & & \qquad \text{in }\Omega _{p} \\
& D_{p}\left( \cdot \right) \dfrac{\partial p}{\partial \nu _{\Gamma }}=-\mu
_{p}\left( \cdot \right) \dfrac{\partial \varphi }{\partial \nu _{\Gamma }}p%
\mathbf{+}h_{e}\left( \cdot ,e\right) -h\left( \cdot ,n\right) p & & \qquad
\text{on }\Gamma \\
& p=0 & & \qquad \text{on }\Gamma _{A} \\
& D_{p}\left( \cdot \right) \dfrac{\partial p}{\partial \nu }=-\mu
_{p}\left( \cdot ,\left\vert \nabla \varphi \right\vert \right) \dfrac{%
\partial \varphi }{\partial \nu }p & & \qquad \text{on }\Gamma _{p}.
\end{align}
\end{subequations}
\begin{definition}
$p\in \mathcal{W}_{p}^{2}$ is called a variational solution to the auxiliary
hole problem (\ref{p_red_stat1}) if
\begin{equation}
a_{p}\left( p,v\right) =L_{p}\left( v\right) \qquad \forall v\in \mathcal{W}%
_{p}^{2}  \label{p_weak}
\end{equation}%
where%
\begin{align*}
& a_{p}\left( p,v\right) =\int_{\Omega _{p}}D_{p}\left( \cdot \right) \nabla
p\cdot \nabla vdx+\int_{\Omega _{p}}\mu _{p}\left( \cdot \right) p\left(
\nabla v\cdot \nabla \varphi \right) dx+\int_{\Gamma }h\left( \cdot \right)
pvd\sigma ,\\
& L_{p}\left( v\right) =\int_{\Gamma }h_{e}\left( \cdot \right) vd\sigma .
\end{align*}
\end{definition}

\noindent Using the same arguments as in Sect.~\ref{sec:aux_elec_pb} we
conclude that:

\begin{itemize}
\item the bilinear form $a_{p}\left( u,v\right) $ is continuous on $\mathcal{%
W}_{p}^{2}\times \mathcal{W}_{p}^{2}$ and%
\begin{equation*}
\left\vert a_{p}\left( u,v\right) \right\vert \leq \left\{ \overline{d}%
_{p}+c\left( q,\Omega _{p}\right) \overline{\mu }_{p}\left\Vert \nabla
\varphi \right\Vert _{L^{q}\left( \Omega _{p}\right) }+c\left( \Omega
_{p}\right) \left\Vert h\right\Vert _{L^{2}\left( \Gamma \right) }\right\}
\left\Vert u\right\Vert _{\mathcal{W}_{p}^{2}}\left\Vert v\right\Vert _{%
\mathcal{W}_{p}^{2}}
\end{equation*}

\item we have%
\begin{equation*}
a_{p}\left( v,v\right) \geq \Lambda _{p}\left\Vert v\right\Vert _{\mathcal{W}%
_{p}^{2}}^{2}\qquad \qquad \qquad \forall v\in \mathcal{W}_{p}^{2}
\end{equation*}%
where%
\begin{equation}
\Lambda _{p}:=\underline{d}_{p}-c\left( q,\Omega _{p}\right) \overline{\mu }%
_{p}\left\Vert \nabla \varphi \right\Vert _{L^{q}\left( \Omega _{p}\right) }
\label{Lambda_p}
\end{equation}

\item the linear form $L_{p}\left( v\right) $ is continuous on $\mathcal{W}%
_{p}^{2}$ and%
\begin{equation*}
\left\vert L_{p}\left( v\right) \right\vert \leq c\left( \Omega_{p}\right)
\left\Vert h_{e}\right\Vert _{L^{2}\left( \Gamma \right) }\left\Vert
v\right\Vert _{\mathcal{W}_{p}^{2}}\qquad \qquad \qquad \forall v\in
\mathcal{W}_{p}^{2}.
\end{equation*}
\end{itemize}

\noindent The above properties allow us to prove the following result.

\begin{lemma}[Auxiliary hole problem]
\label{Lemma Holes}Assume that $\left\{ \Omega ,\Gamma _{D},\Gamma
_{N}\right\} $ is a $q-$admissible triple for some $q\geq 3$; let $\varphi $
be given by Lemma \ref{Lemma Poisson 2}; $D_{p}$, $\mu _{p}\in L^{\infty
}\left( \Omega _{p}\right) $ and satisfying the bounds of Tab.~\ref%
{tab:assumptions}; $h,h_{e}\in L^{2}\left( \Gamma \right) $; $h\geq 0$ a.e.
on $\Gamma $. Then there is a $\delta >0$ such that if $\left\Vert \nabla
\varphi \right\Vert _{L^{q}\left( \Omega _{p}\right) }<\delta $ then problem
(\ref{p_red_stat1}) has a unique variational solution $p$. If in addition $%
h_{e}\geq 0$ a.e. on $\Gamma $, then the solution $p\geq 0$ a.e. in $\Omega
_{p}$.
\end{lemma}

\begin{remark}
The above solution satisfies the estimate
\begin{equation}
\left\Vert p\right\Vert _{\mathcal{W}_{p}^{2}}\leq \frac{c\left( \Omega
_{p}\right) }{\Lambda _{p}}\left\Vert h_{e}\right\Vert _{L^{2}\left( \Gamma
\right) }.  \label{p estimate}
\end{equation}
\end{remark}

\section{The fixed-point map}\label{sec:fixed_point_map}
In this section we collect the various auxiliary problems introduced before to end up
with a functional iteration that allows us to construct the solution of the multiscale
solar cell stationary model described in Sect.~\ref{sec:stationary_multiscale_model}.
\subsection{Preparatory lemmas}
Consider the ball of radius $R>0$ in the Hilbert direct sum
$\mathcal{W}_{n}^{2}\oplus \mathcal{W}_{p}^{2}$
\begin{equation*}
\mathcal{B}_{R}=\left\{ \left( n,p\right) \in \mathcal{W}_{n}^{2}\oplus
\mathcal{W}_{p}^{2}:\left\Vert n\right\Vert _{\mathcal{W}_{n}^{2}}^{2}+\left%
\Vert p\right\Vert _{\mathcal{W}_{p}^{2}}^{2}\leq R^{2}\right\}
\end{equation*}%
and its intersection $\mathcal{B}_{R}^{+}$ with the cone of nonnegative
functions $n\geq 0,p\geq 0$. Note that%
\begin{equation*}
\left( n,p\right) \in \mathcal{B}_{R}\qquad \Longrightarrow \qquad
\left\Vert n\right\Vert _{\mathcal{W}_{n}^{2}}\leq R,\qquad \left\Vert
p\right\Vert _{\mathcal{W}_{p}^{2}}\leq R.
\end{equation*}

\begin{lemma}
\label{Lemma g}Let $g\left( n,p\right) $ be given by (\ref{g})
and $\left( n,p\right) \in \mathcal{B}_{R}$. Then $g\in L^{q}\left( \Omega
\right) $ for $2\leq q\leq 6$ and there exists a constant $c=c\left(
q,\Omega_{n},\Omega_{p}\right) =c\left( q,\Omega ,\Gamma \right) $ such
that%
\begin{equation}\label{eq:g}
\left\Vert g\left( n,p\right) \right\Vert _{L^{q}\left( \Omega \right) }\leq cR.
\end{equation}
\end{lemma}

\begin{proof}
By the Sobolev Imbedding Theorem we have $\mathcal{W}_{n}^{2}\hookrightarrow
L^{q}\left( \Omega_{n}\right) $ and $\mathcal{W}_{p}^{2}\hookrightarrow
L^{q}\left( \Omega_{p}\right) $ for $2\leq q\leq 6$, hence%
\begin{eqnarray*}
\left\Vert g\right\Vert _{L^{q}\left( \Omega \right) }^{q} &=&\left\Vert
n\right\Vert _{L^{q}\left( \Omega_{n}\right) }^{q}+\left\Vert p\right\Vert
_{L^{q}\left( \Omega_{p}\right) }^{q}
\leq c\left( q,\Omega_{n}\right) \left\Vert n\right\Vert _{\mathcal{W}%
_{n}^{2}}^{q}+c\left( q,\Omega_{p}\right) \left\Vert p\right\Vert _{%
\mathcal{W}_{p}^{2}}^{q} \\
&<&c\left( q,\Omega_{n}\right) R^{q}+c\left( q,\Omega_{p}\right) R^{q}
\end{eqnarray*}%
and the assertion follows.
\end{proof}

\begin{lemma}
\label{Lemma np}Let $f\left( n,p\right) $ be given by (\ref{eq:f=np})
and $\left( n,p\right) \in \mathcal{B}_{R}$. Then there exists a constant $%
c=c\left( \Omega ,\Gamma \right) $ such that%
\begin{equation}\label{eq:f}
\left\Vert f\left( n,p\right) \right\Vert _{L^{2}\left( \Gamma \right) }\leq
cR^{2}.
\end{equation}
\end{lemma}

\begin{proof}
Proceeding as for (\ref{v estim}), where $q=4$ and $\Omega $ is replaced by
$\Omega_{n}$ or $\Omega_{p}$, yields
\begin{align*}
& \left\Vert n\right\Vert _{L^{4}\left( \Gamma \right) }\leq c\left( \Omega_{n},\Gamma \right)
\left\Vert n\right\Vert _{\mathcal{W}_{n}^{2}}
& \qquad \forall n\in \mathcal{W}_{n}^{2}, \\
& \left\Vert p\right\Vert _{L^{4}\left( \Gamma \right) }\leq c\left( \Omega_{p},
\Gamma \right) \left\Vert p\right\Vert _{\mathcal{W}_{p}^{2}} &\qquad
\qquad \forall p\in \mathcal{W}_{p}^{2}
\end{align*}%
for suitable constants $c\left( \Omega_{n},\Gamma \right) $ and $c\left(
\Omega_{p},\Gamma \right) $. Then%
\begin{eqnarray*}
\left\Vert np\right\Vert _{L^{2}\left( \Gamma \right) } &\leq &\left\Vert
n\right\Vert _{L^{4}\left( \Gamma \right) }\left\Vert p\right\Vert
_{L^{4}\left( \Gamma \right) } \leq
c\left( \Omega_{n},\Gamma \right) c\left( \Omega_{p},\Gamma \right)
\left\Vert n\right\Vert _{\mathcal{W}_{n}^{2}}\left\Vert p\right\Vert _{%
\mathcal{W}_{p}^{2}} \\
&\leq &c\left( \Omega_{n},\Gamma \right) c\left( \Omega_{p},\Gamma \right)
R^{2}
\end{eqnarray*}%
and the assertion follows since $c\left( \Omega_{n},\Gamma \right) c\left(
\Omega_{p},\Gamma \right) =c\left( \Omega ,\Gamma \right) $.
\end{proof}
Our next aim is to prove the existence of a (unique) solution for the nonlinearly
coupled system of partial differential equations~\eqref{poisson},~\eqref{e_red_stat},~\eqref{n_red_stat} and~\eqref{p_red_stat}. To this end we define a
mapping $\mathbf{K}:\mathcal{B}_{R}^{+}\rightarrow \mathcal{W}_{n}^{2}\oplus
\mathcal{W}_{p}^{2}$ and prove that under suitable conditions it satisfies
the Contraction Mapping Theorem. Given the fixed point $\left( n,p\right) $
of $\mathbf{K}$, the potential $\varphi $ and the exciton concentration $e$
can be recovered as the solutions of the corresponding auxiliary problems.

\subsection{The definition}\label{sec:map_definition}
Let $\left( n,p\right) \in \mathcal{B}_{R}^{+}$. Then, the flow-chart of the map
$\left( n^{\ast },p^{\ast }\right) =\mathbf{K}\left( n,p\right)$
consists of three steps (illustrated in detail below) and
is schematically depicted in Fig.~\ref{fig:gummel_map}. The proposed solution map is a
variant of the classic Gummel iteration that is widely adopted in the
treatment of the Drift-Diffusion and Quantum-Drift-Diffusion
model for inorganic semiconductors. In this context, the Gummel map
has been subject of extensive theoretical and computational investigation,
see~\cite{markowich1986stationary,Jerome:AnalyCharTran,deFalcoSacco2005,deFalcoSacco2009}.
\begin{figure}[h!]
\centering
\includegraphics[width=1.10\textwidth,height=0.5\textwidth]{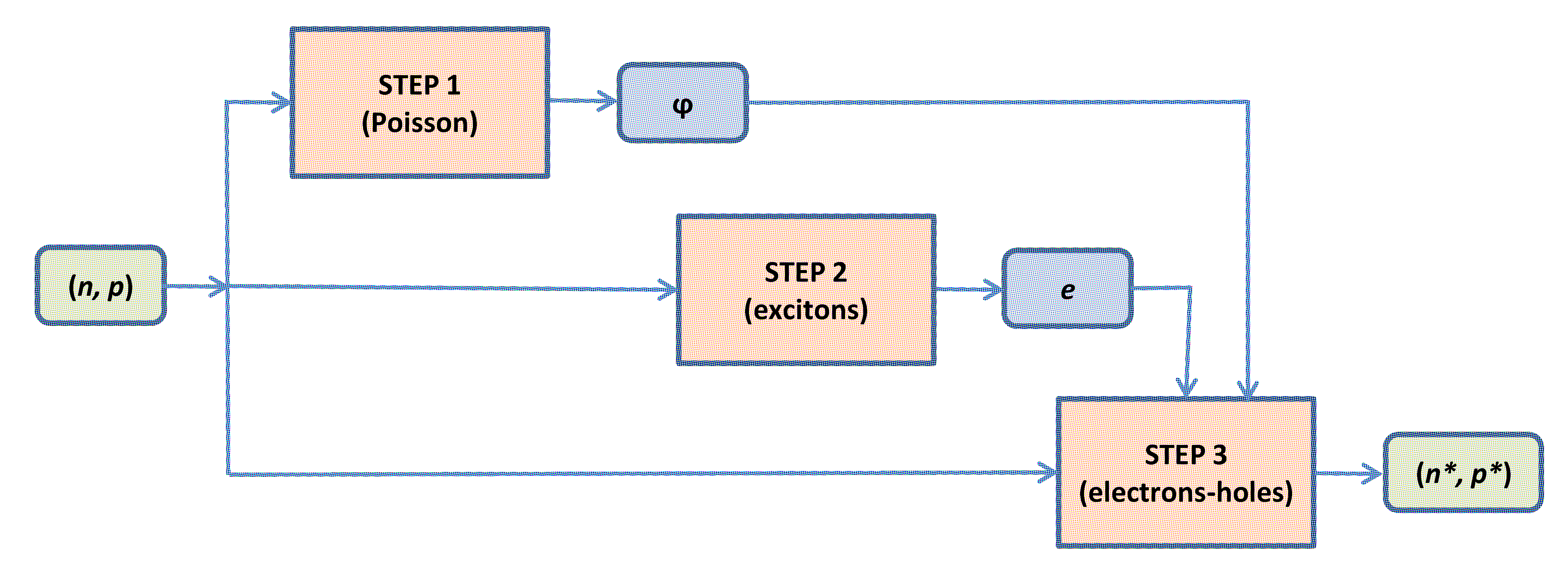}
\caption{Flow-chart of the solution map.}
\label{fig:gummel_map}
\end{figure}
\begin{description}
\item[STEP 1.] Assume that $\left\{ \Omega ,\Gamma _{D},\Gamma
_{N}\right\} $ is a $q-$admissible triple for some $q\geq 3$, $\varepsilon $ as specified in Tab.~\ref{tab:assumptions} and that there exists $\widetilde{\varphi }\in W^{1,q}\left( \Omega \right) $
whose trace on $\partial \Omega $ is equal to $\varphi _{D}$ on $\Gamma_{D}$.
Take $g\left( \cdot \right) \equiv g\left( n\left( \cdot \right) ,p\left(
\cdot \right) \right) $: then $g\left( \cdot \right) \in L^{q}\left( \Omega
\right) $ by Lemma \ref{Lemma g} so that there exists a unique weak solution
$\varphi $ to the auxiliary Poisson problem (\ref{eq:aux poisson}). Moreover
$\varphi \in W^{1,q}\left( \Omega \right) $ and for $i=n,p$,
applying~\eqref{grad fi estim} and~\eqref{eq:g}, we have
\begin{equation}
\left\Vert \nabla \varphi \right\Vert _{L^{q}\left( \Omega _{i}\right) }\leq
c\left( q,\Omega ,\Gamma ,\varepsilon \right) \left( R+\left\Vert \nabla
\widetilde{\varphi }\right\Vert _{L^{q}\left( \Omega \right) }\right).
\label{step 1}
\end{equation}

\item[STEP 2.] Assume $D_{e}$, $Q$, $H$, $\gamma $, $\tau _{e}$, $\tau _{d}$%
, $\eta $, $k_{d}$, $k_{r}$ to be as specified in Tab.~\ref{tab:assumptions}.
Take $f\left( \cdot \right) \equiv n\left( \cdot \right) p\left( \cdot
\right) \geq 0$: then $f\left( \cdot \right) \in L^{2}\left( \Gamma \right) $
by Lemma \ref{Lemma np}. Let $e$ be the unique and nonnegative weak solution
to the auxiliary exciton problem (\ref{e_red_stat1}). Then, using~\eqref{e estimate},~\eqref{eq:f} and the fact that $\overline{\beta }=2\overline{h}\overline{\gamma }\eta$, we get
\begin{equation}
\left\Vert e\right\Vert _{\mathcal{W}^{2}}\leq \frac{c\left( \Omega ,\Gamma
\right) }{\underline{d}_{e}}\left( \left\Vert Q\right\Vert _{L^{2}\left(
\Omega \right) }+\overline{h}\overline{\gamma }\eta R^{2}\right).
\label{step 2}
\end{equation}

\item[STEP 3.] Assume $D_{i}$, $\mu _{i}$, $i=n,p$, to be as specified in
Tab.~\ref{tab:assumptions}. Consider the auxiliary
electron problem~\eqref{n_red_stat1} where we
take $\mu _{n}\left( \cdot \right) \equiv \mu _{n}\left( \cdot ,\left\vert
\nabla \varphi \left( \cdot \right) \right\vert \right) $, $h\left( \cdot
\right) \equiv h\left( \cdot ,p\left( \cdot \right) \right) \geq 0$, $%
h_{e}\left( \cdot \right) \equiv h_{e}\left( \cdot ,e\left( \cdot \right)
\right) \geq 0$. In particular it is $h\left( \cdot \right) ,h_{e}\left(
\cdot \right) \in L^{2}\left( \Gamma \right) $ since $\beta ,\omega \in
L^{\infty }\left( \Gamma \right) $. Let $n^{\ast }$ be the unique and
nonnegative weak solution to (\ref{n_red_stat1}). Then, using~\eqref{n estimate},~\eqref{sup omega} and~\eqref{v estim} we get%
\begin{equation*}
\left\Vert n^{\ast }\right\Vert _{\mathcal{W}_{n}^{2}}\leq \frac{c\left(
\Omega _{n}\right) }{\Lambda _{n}}\left\Vert h_{e}\right\Vert _{L^{2}\left(
\Gamma \right) }\leq \frac{\overline{h}c\left( \Omega _{n}\right) }{\tau
_{d}\Lambda _{n}}\left\Vert e\right\Vert _{L^{2}\left( \Gamma \right) }\leq
\frac{\overline{h}c\left( \Omega _{n}\right) }{\tau _{d}\Lambda _{n}}%
\left\Vert e\right\Vert _{\mathcal{W}^{2}}
\end{equation*}%
and therefore, by~\eqref{step 2},
\begin{equation}
\left\Vert n^{\ast }\right\Vert _{\mathcal{W}_{n}^{2}}\leq c\left( \Omega
,\Gamma \right) \frac{\overline{h}}{\underline{d}_{e}\tau _{d}\Lambda _{n}}%
\left( \left\Vert Q\right\Vert _{L^{2}\left( \Omega \right) }+\overline{h}%
\overline{\gamma }\eta R^{2}\right).   \label{step 3}
\end{equation}%
Applying similar arguments to the auxiliary hole problem~\eqref{p_red_stat1}
we obtain the following estimate for the function $p^{\ast }\geq 0$, unique
and nonnegative weak solution to (\ref{p_red_stat1}):
\begin{equation}
\left\Vert p^{\ast }\right\Vert _{\mathcal{W}_{p}^{2}}\leq c\left( \Omega
,\Gamma \right) \frac{\overline{h}}{\underline{d}_{e}\tau _{d}\Lambda _{p}}%
\left\{ \left\Vert Q\right\Vert _{L^{2}\left( \Omega \right) }+\overline{h}%
\overline{\gamma }\eta R^{2}\right\} .  \label{step 4}
\end{equation}
\end{description}

\begin{remark}
From Lemmas \ref{Lemma Electrons} and \ref{Lemma Holes} we know that $%
n^{\ast }$ and $p^{\ast }$ exist provided $\left\Vert \nabla \varphi
\right\Vert _{L^{q}\left( \Omega_{i}\right) }<\delta $ ($i=n,p$) for a
small enough $\delta $ or, equivalently, if $\Lambda _{n}>0$ and $\Lambda
_{p}>0$.
\end{remark}

\subsection{The invariant set}

In this section we seek a sufficient condition for $\mathbf{K}$ to act
invariantly upon $\mathcal{B}_{R}^{+}$, i.e.%
\begin{equation}
\sqrt{\left\Vert n^{\ast }\right\Vert _{\mathcal{W}_{n}^{2}}^{2}+\left\Vert
p^{\ast }\right\Vert _{\mathcal{W}_{p}^{2}}^{2}}\leq R.
\label{eq:invariant_set_condition}
\end{equation}%
Using~(\ref{step 1}) in~(\ref{Lambda_n}) we get
\begin{equation*}
\Lambda _{n}\geq \underline{d}_{n}-\overline{\mu }_{n}c\left( q,\Omega
,\Gamma ,\varepsilon \right) \left( R+\left\Vert \nabla \widetilde{\varphi }%
\right\Vert _{L^{q}\left( \Omega \right) }\right) .
\end{equation*}%
Set, for notational simplicity,%
\begin{align*}
& c_{0}:=c\left( q,\Omega ,\Gamma ,\varepsilon \right)  \\
& \underline{d}:=\min \left\{ \underline{d}_{n};\underline{d}_{p}\right\}  \\
& \overline{\mu }:=\max \left\{ \overline{\mu }_{n},\overline{\mu }%
_{p}\right\}  \\
& \overline{R}:=\frac{\underline{d}}{c_{0}\overline{\mu }}-\left\Vert \nabla
\widetilde{\varphi }\right\Vert _{L^{q}\left( \Omega \right) }
\end{align*}
Then
\begin{equation}\label{Lambda_n_2}
\Lambda _{n}\geq \underline{d}-c_{0}\overline{\mu }\left( R+\left\Vert
\nabla \widetilde{\varphi }\right\Vert _{L^{q}\left( \Omega \right) }\right)
=c_{0}\overline{\mu }\left( \overline{R}-R\right) .
\end{equation}%
Now assume $0<R<\overline{R}$ (so that $\Lambda _{n}>0$). Then by (\ref{step
3})
\begin{equation*}
\left\Vert n^{\ast }\right\Vert _{\mathcal{W}_{n}^{2}}\leq \frac{c\left(
\Omega ,\Gamma \right) }{c_{0}}\frac{\overline{h}}{\underline{d}_{e}\tau _{d}%
\overline{\mu }}\frac{\left\Vert Q\right\Vert _{L^{2}\left( \Omega \right) }+%
\overline{h}\overline{\gamma }\eta R^{2}}{\overline{R}-R}.
\end{equation*}%
A similar estimate holds for $\left\Vert p^{\ast }\right\Vert _{\mathcal{W}%
_{p}^{2}}$ so that%
\begin{equation*}
\sqrt{\left\Vert n^{\ast }\right\Vert _{\mathcal{W}_{n}^{2}}^{2}+\left\Vert
p^{\ast }\right\Vert _{\mathcal{W}_{p}^{2}}^{2}}\leq c_{1}\frac{\overline{h}%
}{\underline{d}_{e}\tau _{d}\overline{\mu }}\frac{\left\Vert Q\right\Vert
_{L^{2}\left( \Omega \right) }+\overline{h}\overline{\gamma }\eta R^{2}}{%
\overline{R}-R}
\end{equation*}%
where $c_{1}=c_{1}\left( q,\Omega ,\Gamma ,\varepsilon \right) $. To satisfy~%
\eqref{eq:invariant_set_condition} we have to require that
\begin{equation*}
c_{1}\frac{\overline{h}}{\underline{d}_{e}\tau _{d}\overline{\mu }}\frac{%
\left\Vert Q\right\Vert _{L^{2}\left( \Omega \right) }+\overline{h}\overline{%
\gamma }\eta R^{2}}{\overline{R}-R}\leq R.
\end{equation*}%
Write the above inequality as%
\begin{equation}
R^{2}-a\overline{R}R+ab\leq 0  \label{R eq}
\end{equation}%
where
\begin{equation}\label{ab}
a:=\left( 1+c_{1}\frac{\overline{h}^{2}\overline{\gamma }\eta }{\underline{d}%
_{e}\tau _{d}\overline{\mu }}\right) ^{-1}\qquad ;\qquad b:=c_{1}\frac{%
\overline{h}}{\underline{d}_{e}\tau _{d}\overline{\mu }}\left\Vert
Q\right\Vert _{L^{2}\left( \Omega \right) }
\end{equation}%
Inequality~\eqref{R eq} is solvable iff%
\begin{equation*}
\overline{R}\geq \sqrt{\frac{4b}{a}}
\end{equation*}%
with solutions%
\begin{equation}
0<\underset{R_{1}}{\underbrace{\frac{a\overline{R}-\sqrt{a^{2}\overline{R}%
^{2}-4ab}}{2}}}\text{ }\leq R\leq \text{ }\underset{R_{2}}{\underbrace{\frac{%
a\overline{R}+\sqrt{a^{2}\overline{R}^{2}-4ab}}{2}}}.
\label{eq:values_or_R_1_and_R_2}
\end{equation}%
In conclusion, the map $\mathbf{K}$ acts invariantly upon $\mathcal{B}%
_{R}^{+}$ (i.e., $\mathbf{K}\mathcal{B}_{R}^{+}\subset \mathcal{B}_{R}^{+}$)
only for all the values of $R$ satisfying the set of conditions:
\begin{subequations}
\begin{align}
& \overline{R}\geq \sqrt{\dfrac{4b}{a}}  \label{eq:cond_1} \\
& 0<R<\overline{R} \label{eq:cond_2}\\
& R_{1}\leq R\leq R_{2}. \label{eq:cond_3}
\end{align}
\end{subequations}%
Condition~\eqref{eq:cond_1} reads explicitly%
\begin{equation}
\left\Vert \nabla \widetilde{\varphi }\right\Vert _{L^{q}\left( \Omega
\right) }+\sqrt{\dfrac{4c_{1}\overline{h}}{\underline{d}_{e}\tau _{d}%
\overline{\mu }}\left( 1+\dfrac{c_{1}\overline{h}^{2}\overline{\gamma }\eta
}{\underline{d}_{e}\tau _{d}\overline{\mu }}\right) \left\Vert Q\right\Vert
_{L^{2}\left( \Omega \right) }}\leq \dfrac{\underline{d}}{c_{0}\overline{\mu
}}  \label{cond2}
\end{equation}%
so it is certainly satisfied provided that
\begin{subequations}\label{cond2bis}
\begin{align}
& \left\Vert \nabla \widetilde{\varphi }\right\Vert _{L^{q}\left( \Omega
\right) }\text{ and }\dfrac{\overline{h}}{\underline{d}_{e}\tau _{d}%
\overline{\mu }}\left\Vert Q\right\Vert _{L^{2}\left( \Omega \right) }\text{
are small enough, or} & \label{cond2bis_1} \\
& \dfrac{\, \underline{d}\,}{\overline{\mu }}\text{ is large enough.} & \label{cond2bis_2}
\end{align}
\end{subequations}

\noindent From~\eqref{R eq} and the fact that $0<a<1$ it follows that
\begin{equation*}
0<R_{1}\leq R_{2}<R_{1}+R_{2}=a\overline{R}<\overline{R}
\end{equation*}%
which implies that~\eqref{eq:cond_3} is more restrictive than~\eqref{eq:cond_2}. We have thus proved the following result.

\begin{proposition}[Existence of an invariant set for $\mathbf{K}$]
\label{Proposition invariant set} Let $\mathbf{K}:\mathcal{B}%
_{R}^{+}\longrightarrow \mathcal{W}_{n}^{2}\oplus \mathcal{W}_{p}^{2}$ be
the map $\left( n^{\ast },p^{\ast }\right) =\mathbf{K}\left( n,p\right) $
defined through Steps 1-3. Assume that (\ref{cond2bis}) holds. Then $\mathbf{K}%
\mathcal{B}_{R}^{+}\subset \mathcal{B}_{R}^{+}$ for all $R$ satisfying $%
R_{1}\leq R\leq R_{2}$, where the values of $R_{1}$ and $R_{2}$ are given in~%
\eqref{eq:values_or_R_1_and_R_2}.
\end{proposition}

\subsection{Fixed-point by contraction}

The goal of this section is to prove that $\mathbf{K}$ is a strict contraction mapping of
$\mathcal{B}_{R}^{+}$ into itself. This ensures that $\mathbf{K}$ admits a unique fixed point.
Let $\left( n_{i},p_{i}\right) \in
\mathcal{B}_{R}^{+}$ and $\left( n_{i}^{\ast },p_{i}^{\ast }\right) =
\mathbf{K}\left( n_{i},p_{i}\right) $ ($i=1,2$): then we seek a constant $\lambda
\in \left( 0,1\right) $ such that%
\begin{equation}
\sqrt{\left\Vert n_{2}^{\ast }-n_{1}^{\ast }\right\Vert _{\mathcal{W}%
_{n}^{2}}^{2}+\left\Vert p_{2}^{\ast }-p_{1}^{\ast }\right\Vert _{\mathcal{W}%
_{p}^{2}}^{2}}\leq \lambda \sqrt{\left\Vert n_{2}-n_{1}\right\Vert _{%
\mathcal{W}_{n}^{2}}^{2}+\left\Vert p_{2}-p_{1}\right\Vert _{\mathcal{W}%
_{p}^{2}}^{2}}.
\label{contraction}
\end{equation}

\begin{lemma}
\label{Lemma g2-g1}Let $\left( n_{i},p_{i}\right) \in \mathcal{B}_{R}^{+}$
and%
\begin{equation}\label{eq: g2-g1}
g\left( \cdot \right) =g\left( n_{2},p_{2}\right) -g\left(
n_{1},p_{1}\right) =\left\{
\begin{array}{ll}
-\left( n_{2}-n_{1}\right) & \text{in \ }\Omega_{n} \\
p_{2}-p_{1} & \text{in }\Omega_{p}.%
\end{array}%
\right.
\end{equation}%
Then $g\in L^{q}\left( \Omega \right) $ for $2\leq q\leq 6$ and there exists
a constant $c=c\left( q,\Omega ,\Gamma \right) $ such that%
\begin{equation}\label{eq: g2-g1_est}
\left\Vert g\right\Vert _{L^{q}\left( \Omega \right) }\leq c\sqrt{\left\Vert
n_{2}-n_{1}\right\Vert _{\mathcal{W}_{n}^{2}}^{2}+\left\Vert
p_{2}-p_{1}\right\Vert _{\mathcal{W}_{p}^{2}}^{2}}.
\end{equation}
\end{lemma}

\begin{proof}
Proceeding as in the proof of Lemma~\ref{Lemma g} we have%
\begin{eqnarray*}
\left\Vert g\right\Vert _{L^{q}\left( \Omega \right) }^{q} &\leq &c\left(
q,\Omega_{n}\right) \left\Vert n_{2}-n_{1}\right\Vert _{\mathcal{W}%
_{n}^{2}}^{q}+c\left( q,\Omega_{p}\right) \left\Vert p_{2}-p_{1}\right\Vert
_{\mathcal{W}_{p}^{2}}^{q} \\
&\leq &c\left( q,\Omega ,\Gamma \right) \left( \left\Vert
n_{2}-n_{1}\right\Vert _{\mathcal{W}_{n}^{2}}^{q}+\left\Vert
p_{2}-p_{1}\right\Vert _{\mathcal{W}_{p}^{2}}^{q}\right).
\end{eqnarray*}%
Then the assertion follows from the inequality
\begin{equation*}
\left( a^{q}+b^{q}\right) ^{1/q}\leq \left( a^{2}+b^{2}\right) ^{1/2}
\end{equation*}
with $a>0$, $b>0$, $q\geq 2$.
\end{proof}

\begin{lemma}
\label{Lemma n2-n1}Let $\left( n_{i},p_{i}\right) \in \mathcal{B}_{R}^{+}$
and%
\begin{equation}\label{eq:Lemma n2-n1}
f\left( \cdot \right) =n_{2}p_{2}-n_{1}p_{1}
\end{equation}
Then $f\in L^{2}\left( \Gamma \right) $ and there exists a constant $%
c=c\left( \Omega ,\Gamma \right) $ such that%
\begin{equation}\label{eq:Lemma n2-n1_est}
\left\Vert f\right\Vert _{L^{2}\left( \Gamma \right)
}\leq cR\sqrt{\left\Vert n_{2}-n_{1}\right\Vert _{\mathcal{W}%
_{n}^{2}}^{2}+\left\Vert p_{2}-p_{1}\right\Vert _{\mathcal{W}_{p}^{2}}^{2}}
\end{equation}
\end{lemma}

\begin{proof}
We have
\begin{equation*}
n_{2}p_{2}-n_{1}p_{1}=p_{2}\left( n_{2}-n_{1}\right) +n_{1}\left(
p_{2}-p_{1}\right)
\end{equation*}%
so that, by~\eqref{v estim} (where $\mathcal{W}^{2}$ is substituted by $%
\mathcal{W}_{p}^{2}$ and $\mathcal{W}_{n}^{2}$), we obtain
\begin{eqnarray*}
\left\Vert f\right\Vert _{L^{2}\left( \Gamma \right) } &\leq &\left\Vert
p_{2}\left( n_{2}-n_{1}\right) \right\Vert _{L^{2}\left( \Gamma \right)
}+\left\Vert n_{1}\left( p_{2}-p_{1}\right) \right\Vert _{L^{2}\left( \Gamma
\right) } \\
&\leq &\left\Vert p_{2}\right\Vert _{L^{4}\left( \Gamma \right) }\left\Vert
n_{2}-n_{1}\right\Vert _{L^{4}\left( \Gamma \right) }+\left\Vert
n_{1}\right\Vert _{L^{4}\left( \Gamma \right) }\left\Vert
p_{2}-p_{1}\right\Vert _{L^{4}\left( \Gamma \right) } \\
&\leq &c\left( \Omega ,\Gamma \right) \left( \left\Vert p_{2}\right\Vert _{%
\mathcal{W}_{p}^{2}}\left\Vert n_{2}-n_{1}\right\Vert _{\mathcal{W}%
_{n}^{2}}+\left\Vert n_{1}\right\Vert _{\mathcal{W}_{n}^{2}}\left\Vert
p_{2}-p_{1}\right\Vert _{\mathcal{W}_{p}^{2}}\right)  \\
&\leq &c\left( \Omega ,\Gamma \right) R\sqrt{\left\Vert
n_{2}-n_{1}\right\Vert _{\mathcal{W}_{n}^{2}}^{2}+\left\Vert
p_{2}-p_{1}\right\Vert _{\mathcal{W}_{p}^{2}}^{2}}.
\end{eqnarray*}
\end{proof}

Given $\left( n_{i},p_{i}\right) \in \mathcal{B}_{R}^{+}$, $i=1,2$, we call
$\varphi_{i}=u_{i}+\widetilde{\varphi }$ and $e_{i}$ the corresponding functions
computed by Steps 1 and 2, respectively.
Each of the $u_{i}$'s satisfies problem (\ref{Poisson_weak})
so that, taking the difference, we obtain
\begin{equation}
a\left( u_{2}-u_{1},v\right) =\int_{\Omega }\left( g\left(
n_{2},p_{2}\right) -g\left( n_{1},p_{1}\right) \right) v\,dx
\label{Poisson 2-1}
\end{equation}%
Problem~\eqref{Poisson 2-1} looks like the auxiliary Poisson problem (see
Lemma~\ref{Lemma Poisson 2} where $g$ is given by (\ref{eq: g2-g1}) and $\widetilde{\varphi }=0$), hence by
(\ref{u estimate1}) and (\ref{eq: g2-g1_est}) we have%
\begin{equation*}
\left\Vert u_{2}-u_{1}\right\Vert _{\mathcal{W}^{q}}\leq c\left( q,\Omega
,\Gamma ,\varepsilon \right) \sqrt{\left\Vert n_{2}-n_{1}\right\Vert _{%
\mathcal{W}_{n}^{2}}^{2}+\left\Vert p_{2}-p_{1}\right\Vert _{\mathcal{W}%
_{p}^{2}}^{2}}.
\end{equation*}%
In particular: $u_{2}-u_{1}=\left( \varphi _{2}-\widetilde{\varphi }\right)
-\left( \varphi _{1}-\widetilde{\varphi }\right) =\varphi _{2}-\varphi _{1}$
and%
\begin{equation}
\left\Vert \nabla \varphi _{2}-\nabla \varphi _{1}\right\Vert _{L^{q}\left(
\Omega \right) }\leq c\left( q,\Omega ,\Gamma ,\varepsilon \right) \sqrt{%
\left\Vert n_{2}-n_{1}\right\Vert _{\mathcal{W}_{n}^{2}}^{2}+\left\Vert
p_{2}-p_{1}\right\Vert _{\mathcal{W}_{p}^{2}}^{2}}.  \label{estimate fi2-fi1}
\end{equation}
Let us now consider the $e_{i}$'s. Each of them satisfies problem (\ref%
{e_weak}), so that
\begin{equation}
b\left( e_{2}-e_{1},v\right) =\int_{\Gamma }\beta \left( \cdot \right)
\left( n_{2}p_{2}-n_{1}p_{1}\right) v\,d\sigma .  \label{excitons 2-1}
\end{equation}%
Problem (\ref{excitons 2-1}) looks like the auxiliary exciton problem (see
Lemma \ref{Lemma Excitons} where $f$ is given by (\ref{eq:Lemma n2-n1}) and $Q=0$), hence by (%
\ref{e estimate}) and (\ref{eq:Lemma n2-n1_est}) we have%
\begin{equation}
\left\Vert e_{2}-e_{1}\right\Vert _{\mathcal{W}^{2}}\leq \frac{c\left(
\Omega ,\Gamma \right) \overline{\beta }R}{\underline{d}_{e}}\sqrt{%
\left\Vert n_{2}-n_{1}\right\Vert _{\mathcal{W}_{n}^{2}}^{2}+\left\Vert
p_{2}-p_{1}\right\Vert _{\mathcal{W}_{p}^{2}}^{2}}.  \label{estimate e2-e1}
\end{equation}

Now we need the following assumption: the drift velocities $\mu _{i}\left(
\cdot ,E\right) \mathbf{E}$ ($i=n,p$) are Lipschitzian with respect to $%
\mathbf{E}$, namely%
\begin{subequations}\label{Assumption mu}
\begin{equation}
\left\vert \mu _{i}\left( \cdot ,E_{2}\right) \mathbf{E}_{2}-\mu _{i}\left(
\cdot ,E_{1}\right) \mathbf{E}_{1}\right\vert \leq \mu _{i,0}\left( \cdot
\right) \left\vert \mathbf{E}_{2}-\mathbf{E}_{1}\right\vert
\end{equation}%
where%
\begin{equation}
0\leq \mu _{i,0}\left( \cdot \right) \in L^{\infty }\left( \Omega
_{i}\right).
\end{equation}
\end{subequations}

\begin{remark}
Assumption (\ref{Assumption mu}) is trivially satisfied if $\mu
_{i}\left( \cdot ,E\right) \equiv \mu _{i,0}\left( \cdot \right) \in
L^{\infty }\left( \Omega_{i}\right) $. It is also satisfied by the model
proposed in~\cite{MEN2008}, i.e. the functions $\mu _{i}\left( \mathbf{x},\cdot
\right) $ enjoy the conditions stated in Tab.~\ref{tab:assumptions} and, in
addition, they are Lipschitzian ($L_{i}\in L^{\infty }\left( \Omega
_{i}\right) $)%
\begin{equation*}
\left\vert \mu _{i}\left( \cdot ,E_{2}\right) -\mu _{i}\left( \cdot
,E_{1}\right) \right\vert \leq L_{i}\left( \cdot \right) \left\vert
E_{2}-E_{1}\right\vert
\end{equation*}%
and there exists a cutoff field $E^{\ast }$ above which $\mu _{i}\left(
\mathbf{x},\cdot \right) \equiv \mu _{i,0}\left( \mathbf{x}\right) $.
As a matter of fact, for $0\leq E_{1}<E_{2}$, we have
\begin{eqnarray*}
\left\vert \mu _{i}\left( \cdot ,E_{2}\right) \mathbf{E}_{2}-\mu _{i}\left(
\cdot ,E_{1}\right) \mathbf{E}_{1}\right\vert  &=&\left\vert \mu _{i}\left(
\cdot ,E_{2}\right) \left( \mathbf{E}_{2}-\mathbf{E}_{1}\right) +\left( \mu
_{i}\left( \cdot ,E_{2}\right) -\mu _{i}\left( \cdot ,E_{1}\right) \right)
\mathbf{E}_{1}\right\vert  \\
&\leq &\mu _{i}\left( \cdot ,E_{2}\right) \left\vert \mathbf{E}_{2}-\mathbf{E%
}_{1}\right\vert +\left\vert \mu _{i}\left( \cdot ,E_{2}\right) -\mu
_{i}\left( \cdot ,E_{1}\right) \right\vert E_{1}.
\end{eqnarray*}%
Let $E^{\ast }\leq E_{1}$. Then $\mu _{i}\left( \cdot ,E_{1}\right) =\mu
_{i}\left( \cdot ,E_{2}\right) =\mu _{i,0}\left( \cdot \right) $ so that
(\ref{Assumption mu}) are obtained. On the other hand, if $E_{1}\leq E^{\ast }$,
we have%
\begin{eqnarray*}
\left\vert \mu _{i}\left( \cdot ,E_{2}\right) \mathbf{E}_{2}-\mu _{i}\left(
\cdot ,E_{1}\right) \mathbf{E}_{1}\right\vert  &\leq &\overline{\mu }%
_{i}\left\vert \mathbf{E}_{2}-\mathbf{E}_{1}\right\vert +L_{i}\left( \cdot
\right) \left\vert E_{2}-E_{1}\right\vert E^{\ast } \\
&\leq &\left( \overline{\mu }_{i}+L_{i}\left( \cdot \right) \right)
\left\vert \mathbf{E}_{2}-\mathbf{E}_{1}\right\vert
\end{eqnarray*}%
since $\left\vert E_{2}-E_{1}\right\vert \leq \left\vert \mathbf{E}_{2}-\mathbf{E}_{1}\right\vert $,
and (\ref{Assumption mu}) are again obtained.
\end{remark}

Let us now consider the outputs of the solution map, $\left( n_{i}^{\ast
},p_{i}^{\ast }\right) $, $i=1,2$. Each of the $n_{i}^{\ast }$'s satisfies
problem (\ref{n_weak}), so that
\begin{equation*}
\int_{\Omega _{n}}D_{n}\left( \cdot \right) \nabla n_{i}^{\ast }\cdot \nabla
vdx-\int_{\Omega _{n}}\mu _{n}\left( \cdot ,\left\vert \nabla \varphi
_{i}\right\vert \right) n_{i}^{\ast }\nabla \varphi _{i}\cdot \nabla
vdx+\eta ^{-1}\int_{\Gamma }\beta \left( \cdot \right) p_{i}n_{i}^{\ast
}vd\sigma =\int_{\Gamma }\omega \left( \cdot \right) e_{i}vd\sigma
\end{equation*}%
Setting $\mu _{i}\left( \cdot \right) \equiv \mu _{n}\left( \cdot
,\left\vert \nabla \varphi _{i}\right\vert \right) $ for brevity and
subtracting $i=1$ from $i=2$, we obtain
\begin{eqnarray*}
&&\int_{\Omega _{n}}D_{n}\left( \cdot \right) \nabla \left( n_{2}^{\ast
}-n_{1}^{\ast }\right) \cdot \nabla vdx+\eta ^{-1}\int_{\Gamma }\beta \left(
\cdot \right) \left( p_{2}n_{2}^{\ast }-p_{1}n_{1}^{\ast }\right) vd\sigma
\\
&&=\int_{\Omega _{n}}\left( n_{2}^{\ast }\mu _{2}\left( \cdot \right) \nabla
\varphi _{2}-n_{1}^{\ast }\mu _{1}\left( \cdot \right) \nabla \varphi
_{1}\right) \cdot \nabla vdx+\int_{\Gamma }\omega \left( \cdot \right)
\left( e_{2}-e_{1}\right) vd\sigma .
\end{eqnarray*}%
Choose $v=n_{2}^{\ast }-n_{1}^{\ast }$ and use the identity%
\begin{equation*}
p_{2}n_{2}^{\ast }-p_{1}n_{1}^{\ast }=p_{2}\left( n_{2}^{\ast }-n_{1}^{\ast
}\right) +\left( p_{2}-p_{1}\right) n_{1}^{\ast }.
\end{equation*}%
Then%
\begin{eqnarray*}
&&\int_{\Omega _{n}}D_{n}\left( \cdot \right) \left\vert \nabla \left(
n_{2}^{\ast }-n_{1}^{\ast }\right) \right\vert ^{2}dx+\eta ^{-1}\int_{\Gamma
}\beta \left( \cdot \right) p_{2}\left( n_{2}^{\ast }-n_{1}^{\ast }\right)
^{2}d\sigma  \\
&=&\int_{\Omega _{n}}\left( n_{2}^{\ast }\mu _{2}\left( \cdot \right) \nabla
\varphi _{2}-n_{1}^{\ast }\mu _{1}\left( \cdot \right) \nabla \varphi
_{1}\right) \cdot \nabla \left( n_{2}^{\ast }-n_{1}^{\ast }\right) dx \\
&&-\eta ^{-1}\int_{\Gamma }\beta \left( \cdot \right) \left(
p_{2}-p_{1}\right) n_{1}^{\ast }\left( n_{2}^{\ast }-n_{1}^{\ast }\right)
d\sigma +\int_{\Gamma }\omega \left( \cdot \right) \left( e_{2}-e_{1}\right)
\left( n_{2}^{\ast }-n_{1}^{\ast }\right) d\sigma
\end{eqnarray*}%
from which it follows%
\begin{eqnarray}
\underline{d}\left\Vert n_{2}^{\ast }-n_{1}^{\ast }\right\Vert _{\mathcal{W}%
_{n}^{2}}^{2} &\leq &\int_{\Omega _{n}}D_{n}\left( \cdot \right) \left\vert
\nabla \left( n_{2}^{\ast }-n_{1}^{\ast }\right) \right\vert ^{2}dx+\eta
^{-1}\int_{\Gamma }\beta \left( \cdot \right) p_{2}\left( n_{2}^{\ast
}-n_{1}^{\ast }\right) ^{2}d\sigma   \notag \\
&\leq &\int_{\Omega _{n}}\left\vert n_{2}^{\ast }\mu _{2}\left( \cdot
\right) \nabla \varphi _{2}-n_{1}^{\ast }\mu _{1}\left( \cdot \right) \nabla
\varphi _{1}\right\vert \left\vert \nabla \left( n_{2}^{\ast }-n_{1}^{\ast
}\right) \right\vert dx  \notag \\
&&+\eta ^{-1}\overline{\beta }\int_{\Gamma }\left\vert \left(
p_{2}-p_{1}\right) n_{1}^{\ast }\left( n_{2}^{\ast }-n_{1}^{\ast }\right)
\right\vert d\sigma +\overline{\alpha }\int_{\Gamma }\left\vert \left(
e_{2}-e_{1}\right) \left( n_{2}^{\ast }-n_{1}^{\ast }\right) \right\vert
d\sigma   \notag \\
&=&I_{1}+\eta ^{-1}\overline{\beta }I_{2}+\overline{\alpha }I_{3}.
\label{n1}
\end{eqnarray}%
Use of the identity%
\begin{equation*}
n_{2}^{\ast }\mu _{2}\nabla \varphi _{2}-n_{1}^{\ast }\mu _{1}\nabla \varphi
_{1}=n_{2}^{\ast }\left( \mu _{2}\nabla \varphi _{2}-\mu _{1}\nabla \varphi
_{1}\right) +\left( n_{2}^{\ast }-n_{1}^{\ast }\right) \mu _{1}\nabla
\varphi _{1}
\end{equation*}%
gives%
\begin{eqnarray*}
I_{1} &\leq &\int_{\Omega _{n}}\left\vert n_{2}^{\ast }\left( \mu _{2}\nabla
\varphi _{2}-\mu _{1}\nabla \varphi _{1}\right) \right\vert \left\vert
\nabla \left( n_{2}^{\ast }-n_{1}^{\ast }\right) \right\vert dx \\
&+&\int_{\Omega _{n}}\left\vert \left( n_{2}^{\ast }-n_{1}^{\ast }\right)
\mu _{1}\nabla \varphi _{1}\right\vert \left\vert \nabla \left( n_{2}^{\ast
}-n_{1}^{\ast }\right) \right\vert dx=J_{1}+J_{2}
\end{eqnarray*}%
But%
\begin{equation*}
J_{1}\leq \left\Vert n_{2}^{\ast }\right\Vert _{L^{r}\left( \Omega
_{n}\right) }\left\Vert \mu _{2}\nabla \varphi _{2}-\mu _{1}\nabla \varphi
_{1}\right\Vert _{L^{q}\left( \Omega _{n}\right) }\left\Vert \nabla \left(
n_{2}^{\ast }-n_{1}^{\ast }\right) \right\Vert _{L^{2}\left( \Omega
_{n}\right) }
\end{equation*}%
where $1/r+1/q=1/2$ (see~\eqref{3holder}). We have
\begin{equation*}
\left\Vert n_{2}^{\ast }\right\Vert _{L^{r}\left( \Omega _{n}\right) }\leq c\left( q,\Omega _{n}\right)
\left\Vert n_{2}^{\ast }\right\Vert _{\mathcal{W}_{n}^{2}}\leq c\left(
q,\Omega _{n}\right) R
\end{equation*}%
(see~\eqref{emb_est}) and, by (\ref{Assumption mu}) and~\eqref{estimate
fi2-fi1},
\begin{eqnarray*}
\left\Vert \mu _{2}\nabla \varphi _{2}-\mu _{1}\nabla \varphi
_{1}\right\Vert _{L^{q}\left( \Omega _{n}\right) } &\leq &\overline{\mu }%
\left\Vert \nabla \varphi _{2}-\nabla \varphi _{1}\right\Vert _{L^{q}\left(
\Omega _{n}\right) } \\
&\leq &c\left( q,\Omega ,\Gamma ,\varepsilon \right) \overline{\mu }\sqrt{%
\left\Vert n_{2}-n_{1}\right\Vert _{\mathcal{W}_{n}^{2}}^{2}+\left\Vert
p_{2}-p_{1}\right\Vert _{\mathcal{W}_{p}^{2}}^{2}}.
\end{eqnarray*}%
Finally, the application of~\eqref{eq:norm_on_Wiq} with $i=n$ and $q=2$ gives%
\[
\left\Vert \nabla \left( n_{2}^{\ast }-n_{1}^{\ast }\right) \right\Vert
_{L^{2}\left( \Omega _{n}\right) }=\left\Vert n_{2}^{\ast }-n_{1}^{\ast
}\right\Vert _{\mathcal{W}_{n}^{2}}.
\]
Collecting the above estimates, we obtain
\begin{equation*}
J_{1}\leq c\left( q,\Omega ,\Gamma ,\varepsilon \right) \overline{\mu }R%
\sqrt{\left\Vert n_{2}-n_{1}\right\Vert _{\mathcal{W}_{n}^{2}}^{2}+\left%
\Vert p_{2}-p_{1}\right\Vert _{\mathcal{W}_{p}^{2}}^{2}}\left\Vert
n_{2}^{\ast }-n_{1}^{\ast }\right\Vert _{\mathcal{W}_{n}^{2}}.
\end{equation*}
Similarly,%
\begin{eqnarray*}
J_{2} &\leq &\overline{\mu }\int_{\Omega _{n}}\left\vert \left( n_{2}^{\ast
}-n_{1}^{\ast }\right) \nabla \varphi _{1}\right\vert \left\vert \nabla
\left( n_{2}^{\ast }-n_{1}^{\ast }\right) \right\vert dx \\
&\leq &\overline{\mu }\left\Vert n_{2}^{\ast }-n_{1}^{\ast }\right\Vert
_{L^{r}\left( \Omega _{n}\right) }\left\Vert \nabla \varphi _{1}\right\Vert
_{L^{q}\left( \Omega _{n}\right) }\left\Vert \nabla \left( n_{2}^{\ast
}-n_{1}^{\ast }\right) \right\Vert _{L^{2}\left( \Omega _{n}\right) } \\
&\leq &\overline{\mu }c\left( q,\Omega _{n}\right) \left\Vert \nabla \varphi
_{1}\right\Vert _{L^{q}\left( \Omega _{n}\right) }\left\Vert n_{2}^{\ast
}-n_{1}^{\ast }\right\Vert _{\mathcal{W}_{n}^{2}}^{2}.
\end{eqnarray*}
Moreover, we have
\begin{eqnarray*}
I_{2} &=&\int_{\Gamma }\left\vert \left( p_{2}-p_{1}\right) n_{1}^{\ast
}\left( n_{2}^{\ast }-n_{1}^{\ast }\right) \right\vert d\sigma  \\
&\leq &\left\Vert n_{1}^{\ast }\right\Vert _{L^{4}\left( \Gamma \right)
}\left\Vert p_{2}-p_{1}\right\Vert _{L^{4}\left( \Gamma \right) }\left\Vert
n_{2}^{\ast }-n_{1}^{\ast }\right\Vert _{L^{2}\left( \Gamma \right) } \\
&\leq &\left\Vert n_{1}^{\ast }\right\Vert _{L^{4}\left( \partial \Omega
_{n}\right) }\left\Vert p_{2}-p_{1}\right\Vert _{L^{4}\left( \partial \Omega
_{p}\right) }\left\Vert n_{2}^{\ast }-n_{1}^{\ast }\right\Vert _{L^{2}\left(
\partial \Omega _{n}\right) } \\
&\leq &c\left( \Omega ,\Gamma \right) \left\Vert n_{1}^{\ast }\right\Vert
_{H^{1}\left( \Omega _{n}\right) }\left\Vert p_{2}-p_{1}\right\Vert
_{H^{1}\left( \Omega _{p}\right) }\left\Vert n_{2}^{\ast }-n_{1}^{\ast
}\right\Vert _{H^{1}\left( \Omega _{n}\right) } \\
&\leq &c\left( \Omega ,\Gamma \right) \left\Vert n_{1}^{\ast }\right\Vert _{%
\mathcal{W}_{n}^{2}}\left\Vert p_{2}-p_{1}\right\Vert _{\mathcal{W}%
_{p}^{2}}\left\Vert n_{2}^{\ast }-n_{1}^{\ast }\right\Vert _{\mathcal{W}%
_{n}^{2}} \\
&\leq &c\left( \Omega ,\Gamma \right) R\sqrt{\left\Vert
n_{2}-n_{1}\right\Vert _{\mathcal{W}_{n}^{2}}^{2}+\left\Vert
p_{2}-p_{1}\right\Vert _{\mathcal{W}_{p}^{2}}^{2}}\left\Vert n_{2}^{\ast
}-n_{1}^{\ast }\right\Vert _{\mathcal{W}_{n}^{2}}
\end{eqnarray*}%
and, using~(\ref{sup omega}),~(\ref{v estim}) and~(\ref{estimate e2-e1}),
\begin{eqnarray*}
I_{3} &=&\int_{\Gamma }\left\vert \left( e_{2}-e_{1}\right) \left(
n_{2}^{\ast }-n_{1}^{\ast }\right) \right\vert d\sigma \leq
\left\Vert e_{2}-e_{1}\right\Vert _{L^{2}\left( \Gamma \right)
}\left\Vert n_{2}^{\ast }-n_{1}^{\ast }\right\Vert _{L^{2}\left( \Gamma
\right) } \\
&\leq &c\left( \Omega ,\Gamma \right) \left\Vert e_{2}-e_{1}\right\Vert _{%
\mathcal{W}^{2}}\left\Vert n_{2}^{\ast }-n_{1}^{\ast }\right\Vert _{\mathcal{%
W}_{n}^{2}} \\
&\leq &\frac{c\left( \Omega ,\Gamma \right) \overline{\beta }R}{\underline{d}%
_{e}}\sqrt{\left\Vert n_{2}-n_{1}\right\Vert _{\mathcal{W}%
_{n}^{2}}^{2}+\left\Vert p_{2}-p_{1}\right\Vert _{\mathcal{W}_{p}^{2}}^{2}}%
\left\Vert n_{2}^{\ast }-n_{1}^{\ast }\right\Vert _{\mathcal{W}_{n}^{2}}.
\end{eqnarray*}
Inserting the obtained estimates for $I_{1}$, $I_{2}$ and $I_{3}$ into (\ref%
{n1}) yields%
\begin{align*}
& \left( \underline{d}-c\left( q,\Omega ,\Gamma ,\varepsilon \right)
\overline{\mu }\left\Vert \nabla \varphi _{1}\right\Vert _{L^{q}\left(
\Omega _{n}\right) }\right) \left\Vert n_{2}^{\ast }-n_{1}^{\ast
}\right\Vert _{\mathcal{W}_{n}^{2}} \\
& \leq \left( \overline{\mu }c\left( q,\Omega ,\Gamma ,\varepsilon \right)
+\left( \frac{\overline{\beta }}{\eta }+\frac{\overline{\alpha }\overline{%
\beta }}{\underline{d}_{e}}\right) c\left( \Omega ,\Gamma \right) \right) R%
\sqrt{\left\Vert n_{2}-n_{1}\right\Vert _{\mathcal{W}_{n}^{2}}^{2}+\left%
\Vert p_{2}-p_{1}\right\Vert _{\mathcal{W}_{p}^{2}}^{2}}
\end{align*}%
But using~(\ref{step 1}) gives
\begin{equation*}
\underline{d}-\overline{\mu }c\left( q,\Omega ,\Gamma ,\varepsilon \right)
\left\Vert \nabla \varphi _{1}\right\Vert _{L^{q}\left( \Omega _{n}\right)
}\geq \underline{d}-\overline{\mu }c\left( q,\Omega ,\Gamma ,\varepsilon
\right) \left( R+\left\Vert \nabla \widetilde{\varphi }\right\Vert
_{L^{q}\left( \Omega \right) }\right) =c_{0}\overline{\mu }\left( \overline{R%
}-R\right)
\end{equation*}%
where $c_{0}$ can be chosen as in (\ref{Lambda_n}) without loss of generality. Then%
\begin{equation*}
\left\Vert n_{2}^{\ast }-n_{1}^{\ast }\right\Vert _{\mathcal{W}_{n}^{2}}\leq
\left( 1+\left( \frac{\overline{\beta }}{\overline{\mu }\eta }+\frac{%
\overline{\alpha }\overline{\beta }}{\overline{\mu }\underline{d}_{e}}%
\right) c_{1}\right) \frac{R}{\overline{R}-R}\sqrt{\left\Vert
n_{2}-n_{1}\right\Vert _{\mathcal{W}_{n}^{2}}^{2}+\left\Vert
p_{2}-p_{1}\right\Vert _{\mathcal{W}_{p}^{2}}^{2}}
\end{equation*}%
where $c_{1}$ can be chosen as in (\ref{ab}). A similar estimate can be proved
to hold also for $\left\Vert p_{2}^{\ast }-p_{1}^{\ast }\right\Vert _{%
\mathcal{W}_{p}^{2}}$, so that, in conclusion, we get (\ref{contraction}%
) where%
\begin{equation*}
\lambda =\frac{\widehat{c}R}{\overline{R}-R}
\end{equation*}%
having set
\begin{equation*}
\widehat{c}:=\sqrt{2}\left( 1+c_{1}\frac{\overline{\beta }}{\overline{\mu }}%
\left( \frac{1}{\eta }+\frac{\overline{\alpha }}{\underline{d}_{e}}\right)
\right).
\end{equation*}%
By Proposition \ref{Proposition invariant set} we know that $0<R_{1} \leq R \leq R_{2}<%
\overline{R}$. Now, it is $0<\lambda <1$ if and only if%
\begin{equation*}
0<R<\frac{\overline{R}}{1+\widehat{c}}<\overline{R}
\end{equation*}%
so that the map $\mathbf{K}$ is a contraction provided that
\begin{equation*}
R_{1}<\frac{\overline{R}}{1+\widehat{c}}
\end{equation*}%
i.e.%
\begin{equation}
\underset{=\text{ }\widehat{a}}{\underbrace{\left( a-\frac{2}{1+\widehat{c}}%
\right) }}\,\overline{R}<\sqrt{a^{2}\overline{R}^{2}-4ab}.  \label{a^}
\end{equation}%
Two cases are in order:

\begin{itemize}
\item if $\widehat{a}\leq 0$, then condition (\ref{a^}) is satisfied, hence
under condition (\ref{cond2}) the map $\mathbf{K}:\mathcal{B}%
_{R}^{+}\longrightarrow \mathcal{B}_{R}^{+}$ is a contraction for all $R$
satisfying $R_{1} \leq R<\min \left\{ R_{2};\frac{\overline{R}}{1+\widehat{c}}%
\right\} $;

\item if $\widehat{a}>0$, then condition (\ref{a^}) reads%
\begin{equation}
\overline{R}>\sqrt{\frac{4ab}{a^{2}-\widehat{a}^{2}}}=\frac{1}{\sqrt{%
1-\left( \dfrac{\widehat{a}}{a}\right) ^{2}}}\sqrt{\dfrac{4b}{a}}
\label{cond3}
\end{equation}%
which is stronger than condition (\ref{eq:cond_1}) since $\widehat{a}<a$.
Moreover, in this case it is always $\dfrac{\overline{R}}{1+\widehat{c}}%
<R_{2}$. Therefore: under condition (\ref{cond3}) the map $\mathbf{K}:%
\mathcal{B}_{R}^{+}\longrightarrow \mathcal{B}_{R}^{+}$ is a contraction for
all $R$ satisfying $R_{1}\leq R<\frac{\overline{R}}{1+\widehat{c}}$.
\end{itemize}

\noindent Condition (\ref{cond3}) can be written as $\overline{R}>\sigma \sqrt{\dfrac{%
4b}{a}}$ with $\sigma >1$ and it reads explicitly%
\begin{equation}
\left\Vert \nabla \widetilde{\varphi }\right\Vert _{L^{q}\left( \Omega
\right) }+\sigma \sqrt{\dfrac{4c_{1}\overline{h}}{\underline{d}_{e}\tau _{d}%
\overline{\mu }}\left( 1+\dfrac{c_{1}\overline{h}^{2}\overline{\gamma }\eta
}{\underline{d}_{e}\tau _{d}\overline{\mu }}\right) \left\Vert Q\right\Vert
_{L^{2}\left( \Omega \right) }}<\dfrac{\underline{d}}{c_{0}\overline{\mu }}
\label{cond2ter}
\end{equation}%
to be compared with (\ref{cond2}). Both (\ref{cond2}) and (\ref{cond2ter})
are satisfied under condition (\ref{cond2bis}).

We have thus proved the following result.

\begin{theorem}
\label{Main Theorem} Let $\mathbf{K}:\mathcal{B}_{R}^{+}\longrightarrow
\mathcal{W}_{n}^{2}\oplus \mathcal{W}_{p}^{2}$ be
the map $\left( n^{\ast },p^{\ast }\right) =\mathbf{K}\left( n,p\right) $
defined through Steps 1-3. In addition, assume that \eqref{cond2bis} and \eqref{Assumption mu}
hold. Then there exist $R_{2}>R_{1}>0$ such that $\mathbf{K}\mathcal{\ }$is a strict
contraction on $\mathcal{B}_{R}^{+}$ for all $R$ satisfying $R_{1}<R<R_{2}$.
Thus, $\mathbf{K}$ admits a unique fixed point in $\mathcal{B}_{R}^{+}$.
\end{theorem}

\section{Solution map validation through numerical simulation}\label{sec:numerical_simulations}
In this section we carry out a computational validation of the theoretical
properties of the fixed-point map introduced and analyzed in Sect.~\ref{sec:fixed_point_map}.
To this purpose, we consider the realistic three-dimensional solar cell geometry
shown in Fig.~\ref{fig:simulation_domain} which represents the unit cell of an ideal
lattice of chessboard-shaped nanostructures of donor and acceptor materials.
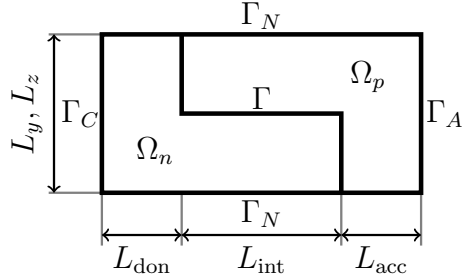
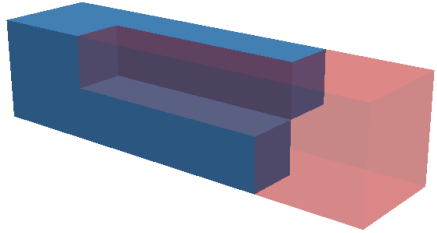
\begin{figure}[h!]
	\centering
	\subfigure[]{
		\scalebox{1}{\begin{tikzpicture}[scale=1.4]

	\draw[line width=0.3mm, gray] (0, -0.5) -- (0, 0);
	\draw[line width=0.3mm, gray] (0.75, -0.5) -- (0.75, 0);
	\draw[line width=0.3mm, gray] (2.25, -0.5) -- (2.25, 0);
	\draw[line width=0.3mm, gray] (3, -0.5) -- (3, 0);
	\draw[line width=0.3mm, gray] (-0.5, 0) -- (0, 0);
	\draw[line width=0.3mm, gray] (-0.5, 1.5) -- (0, 1.5);
	
	\draw[line width=0.6mm] (0,0) rectangle (3,1.5);
	\draw[line width=0.6mm] (0.75, 1.5) -- (0.75, 0.75)
									 	-- (2.25, 0.75)
									 	-- (2.25, 0);
										
	\draw[<->,thick, line width=0.3mm] (0, -0.4) -- (0.75, -0.4) node[sloped,midway,below] {$L_{\mathrm{don}}$};
	\draw[<->,thick, line width=0.3mm] (0.75, -0.4) -- (2.25, -0.4) node[sloped,midway,below] {$L_{\mathrm{int}}$};	\draw[<->,thick, line width=0.3mm] (2.25, -0.4) -- (3, -0.4) node[sloped,midway,below] {$L_{\mathrm{acc}}$};
    \draw[<->,thick, line width=0.3mm] (-0.45, 0) -- (-0.45, 1.5) node[sloped,midway,above] {$L_y, L_z$};
    \draw (0.5, 0.4) node {$\Omega_{n}$};
	\draw (2.5, 1.1) node {$\Omega_{p}$};
	\draw (1.5, 0.9) node {$\Gamma$};
	\draw (-0.2, 0.75) node {$\Gamma_{C}$};
	\draw (3.225, 0.75) node {$\Gamma_{A}$};
	\draw (1.5, 1.675) node {$\Gamma_{N}$};
	\draw (1.5, -0.2) node {$\Gamma_{N}$};
x\
\end{tikzpicture}}
		\label{fig:simulation_domain_schema}}
	\subfigure[]{
		\includegraphics[width=0.45\textwidth]{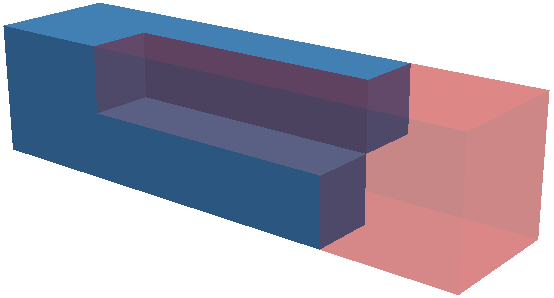}
		\label{fig:simulation_domain_3d}}
	\caption{Simulation domain scheme (a) and 3D representation (b).}
	\label{fig:simulation_domain}
\end{figure}
The whole cell domain is obtained by symmetrically repeating the module of Fig.~\ref{fig:simulation_domain_3d}
with respect to the lateral faces on the part of the boundary
denoted with $\Gamma_N = \partial\Omega \setminus (\Gamma_C \cup \Gamma_A)$
in Fig.~\ref{fig:simulation_domain_schema}. Since $\Omega$ is a convex polyhedron
and the border between $\Gamma_D$ and $\Gamma_N$ consists of a finite number of segments,
then the triple $\left\{ \Omega ,\Gamma _{D},\Gamma _{N}\right\}$
associated with the geometry depicted in Fig.~\ref{fig:simulation_domain}
is $q-$admissible for a $q >3$ (see~\cite{Dauge1992,HALLER2008,Hieber2008}).

The electrochemical behaviour of the cell can be described by
equations~\eqref{eq:excitons}-\eqref{eq:potential}
enforcing symmetry conditions on $\Gamma_N$, i.e.\ applying
zero-flux conditions.
The iterative map illustrated in Sect.~\ref{sec:fixed_point_map} is then applied
to the model equations that are numerically solved upon using a suitable finite element
discretization scheme.

According to the theoretical results of
Theorem~\ref{Main Theorem},
the map should converge to a unique fixed point
if conditions~\eqref{cond2bis} and~\eqref{Assumption mu} are met. In particular we analyzed the
behaviour of the map by systematically changing the value of:
\begin{itemize}
    \item the exciton photogeneration term $Q$;
    \item the electron and hole mobilities $\mu_n$ and $\mu_p$;
    \item the voltage $\varphi_C - \varphi_A$ applied to the electrodes.
\end{itemize}
Conditions~\eqref{cond2bis} suggest that there might be particular threshold values
for such parameters above or below which the convergence of the map is
compromised. We want to investigate whether such values exist and
to study the dependence of the convergence speed of the map on
the value
attained by such parameters.

The following assumptions are made on the functional form of the model parameters:
\begin{enumerate}
\item The light absorption is uniform through the entire cell,  i.e. $Q(%
\mathbf{x}) = Q$;

\item \label{ass:mobilities} Electron and hole mobility parameters $\mu _{n}$
and $\mu _{p}$ depend on the local electric field according to the
functional form ($i=n,p$)%
\begin{equation}
\mu _{i}\left( E\right) =\left\{
\begin{tabular}{lll}
$\mu _{i,0}\exp \left( \dfrac{\beta _{i}\sqrt{E}}{k_{B}T}\right) $ & if & $%
E<E^{\ast }$ \\
$\mu _{i,0}\exp \left( \dfrac{\beta _{i}\sqrt{E^{\ast }}}{k_{B}T}\right) $ &
if & $E\geq E^{\ast }$%
\end{tabular}%
\right.   \label{eq:mob_model}
\end{equation}%
where $E$ is the electric field intensity, $\mu _{i,0}$ is the zero-field
mobility of the charge carrier, $\beta _{i}$ is a modulation parameter, $%
K_{B}$ is the Boltzmann constant and $T$ is the temperature. Definition~%
\eqref{eq:mob_model} is a modified version of a model widely used in the
literature, see e.g.~\cite%
{barker2003,buxton2007,hwang_greenham,williams_th,williams2008}, where the
ceiling for values of $E$ above the cutoff level $E^{\ast }$ has been
introduced to be consistent with the assumptions reported in Tab.~\ref%
{tab:assumptions}.

\item Electron and hole diffusion coefficients have been assumed to be
constant consistently with the assumptions of Tab.~\ref{tab:assumptions} and
given by
\begin{equation}
D_{i}=\dfrac{k_{B}T}{q}\ \dfrac{\mu _{i}\left( 0\right) +\mu _{i}\left(
E^{\ast }\right) }{2}  \label{eq:diffusion_parameters}
\end{equation}%
where $q$ denotes the elementary electric charge. The second term at the
right-hand side of~\eqref{eq:diffusion_parameters} represents the average
between the value of the mobility at zero electric field and that at the
cutoff level $E^{\ast }$ introduced in~\eqref{eq:mob_model}.

\item The bimolecular recombination rate $\gamma $ is defined with the
formula described in \cite{porro2014,barker2003} with the dependence on the
electric field removed to be compliant with the assumption made in Tab.~\ref%
{tab:assumptions}, i.e.
\begin{equation}
\gamma =\dfrac{q}{\varepsilon ^{\ast }}\min \left\{ \mu _{n,0};\mu
_{p,0}\right\}   \label{eq:gammamodel}
\end{equation}%
where $\varepsilon ^{\ast }$ is defined as the harmonic average of the
dielectric permittivities of the acceptor and donor materials
\begin{equation}
\varepsilon ^{\ast }=\left( \dfrac{\varepsilon _{\mathrm{acc}%
}^{-1}+\varepsilon _{\mathrm{don}}^{-1}}{2}\right) ^{-1}.
\end{equation}
\end{enumerate}

In Tab.~\ref{tab:simulation_paramters} we provide a list of the values of
model parameters values used in the simulations. Numbers are in agreement with realistic
data in solar cell modeling and design (see~\cite{barker2003,buxton2007,hwang_greenham,williams_th,williams2008}).

\begin{table}[h!]
	\centering
	\begin{minipage}[t]{.4\linewidth}
		\begin{tabular}[t]{|l|l|l|}
			\hline
			\textsf{parameter} & \textsf{value} \\ \hline\hline
			$L_\mathrm{acc}$ & 25\,nm \\ \hline
			$L_\mathrm{int}$ & 50\,nm \\ \hline
			$L_\mathrm{don}$ & 25\,nm \\ \hline
			$L_\mathrm{y}$ & 25\,nm \\ \hline
			$L_\mathrm{z}$ & 25\,nm \\ \hline
			$\varepsilon_\mathrm{acc}$ & $4\varepsilon_0$ \\ \hline
			$\varepsilon_\mathrm{don}$ & $4\varepsilon_0$ \\ \hline
			$\varphi_C - \varphi_A$ & 0.4\,V or 0\,V  \\ \hline
			$D_e$ & $100\cdot 10^{-9}$\,m$^{2}$s$^{-1}$\\ \hline
			$Q$ & $10^{28}$\,m$^{-3}$s$^{-1}$\\ \hline
            $T$ & 298.16 \, K\\ \hline
		\end{tabular}
	\end{minipage}
	\begin{minipage}[t]{.4\linewidth}
		\begin{tabular}[t]{|l|l|l|}
			\hline
			\textsf{parameter} & \textsf{value} \\ \hline\hline
			$H$ & 1\,nm \\ \hline
			$\mu_{n,0}$ & $300\cdot 10^{-9}$\,m$^{2}$V$^{-1}$s$^{-1}$\\ \hline
			$\mu_{p,0}$ & $100\cdot 10^{-9}$\,m$^{2}$V$^{-1}$s$^{-1}$\\ \hline
			$\beta_{n}$ & $3 \cdot 10^{-4} K_B T$\,V$^{1/2}$m$^{-1/2}$\\ \hline
            $\beta_{p}$ & $3 \cdot 10^{-4} K_B T$\,V$^{1/2}$m$^{-1/2}$\\ \hline
            $E^{*}$ & $10^{7}$\,V\,m$^{-1}$\\ \hline
            $\tau _{d}$ & 1\,ps \\ \hline
			$\tau _{e}$ & 1\,ns \\ \hline
			$k_{d}$ & $1 \cdot 10^{9}$\,s$^{-1}$ \\ \hline
			$k_{r}$ & $0.1 \cdot 10^{9}$\,s$^{-1}$ \\ \hline
			$\eta $ & 0.25 \\ \hline
		\end{tabular}	
	\end{minipage}
	\caption{Model parameter values used in the performed simulations.}
	\label{tab:simulation_paramters}
\end{table}

For the spatial discretization of the PDE
system~\eqref{eq:excitons}-\eqref{eq:potential}
we adopt the Galerkin Finite Element Method
stabilized by an Exponential Fitting technique
(see~
\cite{deFalcoSacco2010,porro2014,bank1998,gatti1998,xu1999,cogliati2010,porro2014,Porro2014phDthesis})
implemented in the Octave package \texttt{bim}\cite{bim}
and we use the software GMSH\cite{gmsh}
and the Octave package \texttt{msh}\cite{msh}
to generate the triangulation of
the computational domain into an unstructured mesh with local refinements
in the regions close to the interface $\Gamma$.

In the implementation of the code we followed the structure
of the map presented in Sect.~\ref{sec:map_definition}
and we used the following stopping criterion on the $H^1$-norm of the
increments of the numerical solutions $n_h$ and $p_h$
\begin{equation}
\left\Vert
n_h^k - n_h^{k-1}
\right\Vert
_{H^{1}\left( \Omega_{n}\right) }
+
\left\Vert
p_h^k - p_h^{k-1}
\right\Vert
_{H^{1}\left( \Omega_{p}\right) }
< \varepsilon \qquad k \geq 1,
\label{eq:stopping_criterion}
\end{equation}
where $\varepsilon$ is the tolerance and $(\cdot)^k$ indicates the solution
obtained at the $k$-th iteration of the map.
In all the simulations we set $\varepsilon=10^{-9}$ and
\begin{equation}
n_h^0(\mathbf{x}) = 0, \
\mathbf{x} \in \Omega_{n}
\quad \text{and} \quad
p_h^0(\mathbf{x}) = 0, \
\mathbf{x} \in \Omega_{p}.
\end{equation}

\subsection{Changing the exciton generation rate $Q$}
\label{sec:num_changing_Q}
Conditions~\eqref{cond2bis} for the contractivity of the
map defined in Sect.~\ref{sec:map_definition}
state that the exciton generation rate term $Q$ has to be
\emph{small enough}, i.e.\ there is an upper limit for it
above which Theorem~\ref{Main Theorem} does not hold and the map is not guaranteed to converge to a
unique point.
Thus, progressively increasing the value of $Q$ we expect the
map to perform less and less efficiently. This means that we expect
the parameter $\lambda$ in (\ref{contraction}) to increase when approaching
the upper limit $\lambda=1$, or, equivalently, the number of iterations to
satisfy~\eqref{eq:stopping_criterion} to increase.

We consider two configurations of applied voltage that correspond to different operation modes of the solar cell:
\begin{itemize}
    \item[1] $\varphi_C - \varphi_A$ = 0.4\,V. A potential
    difference exists between the electrodes,
    whether due to the difference in work
    function of the materials or to some voltage
    applied externally.
    This configuration is representative of the typical operation mode
    of a solar cell generating electric current.

    \item[2] $\varphi_C - \varphi_A$ = 0\,V. The two electrodes
    are at the same potential. This configuration represents
    a suboptimal operation mode as the electric field in the
    cell is small, and the electric charges are not collected efficiently
    at the electrodes.
\end{itemize}

The two configurations are interesting for the analysis
of the performance of the iterative map because the different electric field
profiles, which are determined by the applied electric potential,
result in significantly different profiles for the charge carrier densities.
As the electric field in the cell is almost negligible in configuration~2,
electrons and holes move slowly towards the electrodes
and their density is expected to be high
at the interface between the donor and acceptor
materials where they are generated. As a consequence the bimolecular
second order term $2H\gamma n p$ in
Eqs.~\eqref{eq:gamma_electrons},~\eqref{eq:gamma_holes}
and~\eqref{eq:gamma_polarons}
is expected to be large and to determine a reduction of the performance of the
iterative map.

In Fig.~\ref{fig:convergence_generation} we report the number of iterations
needed by the map to converge to the fixed point in the two configurations
for increasing values of the exciton generation rate $Q$.
The results are in line with our expectations as the number of iterations
increases with $Q$, and in both cases there seems to exist a
specific threshold value such that when $Q$
approaches it the number of iterations increases exponentially
until no convergence is achieved. It is interesting to notice
that such limit value is lower in the configuration with no applied potential,
as a consequence of the different characteristics of such operation mode
discussed in the previous paragraph.
\begin{figure}[h!]
    \centering
    \includegraphics[width=0.5\textwidth]{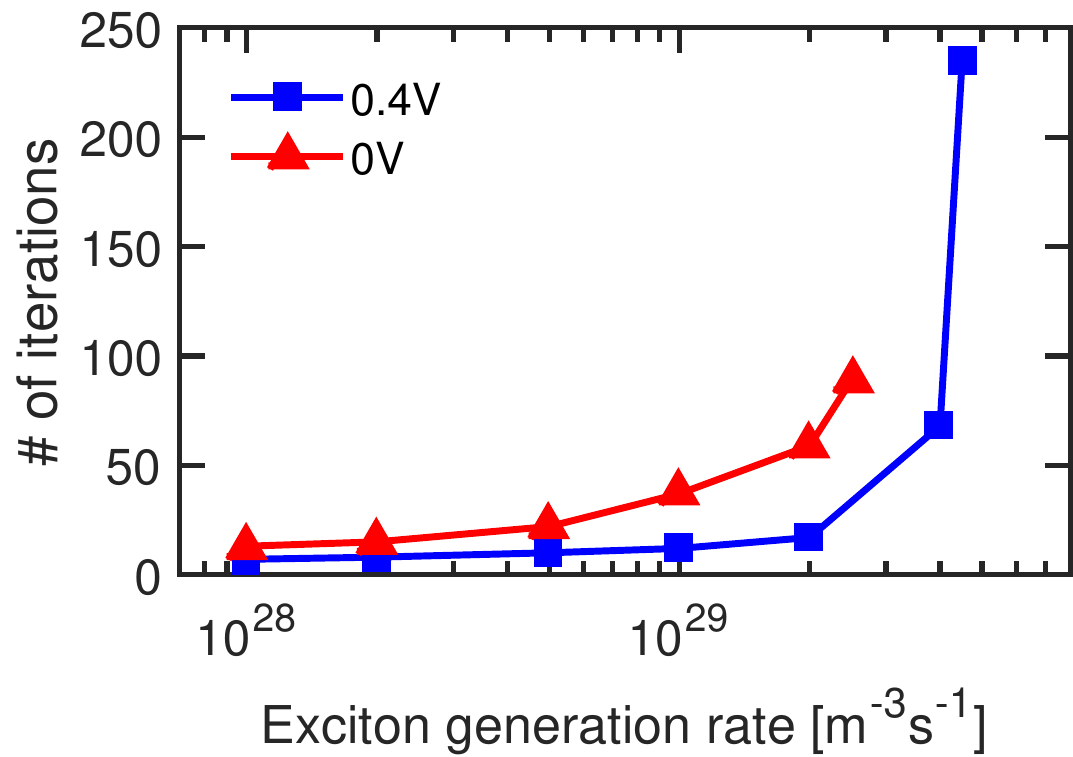}
    \caption{Number of iterations needed by the map to converge
    	changing the value of the exciton generation rate.}
    \label{fig:convergence_generation}
\end{figure}

\subsection{Changing the zero-field charge carrier mobility}
\label{sec:num_changing_mu}
Contractivity conditions~\eqref{cond2bis} depend also on the parameter $\overline{\mu}$,
the maximum between the largest values of electron and hole mobilities, but,
unlike the case of parameter $Q$ in Sect.~\ref{sec:num_changing_Q}, the role of
$\overline{\mu}$ is far less immediate to characterize because it appears
in \emph{both}~\eqref{cond2bis_1} and~\eqref{cond2bis_2}. Moreover,
based on the mechanism of photogenerated
charge transport in the acceptor and donor materials, we expect,
similarly to Sect.~\ref{sec:num_changing_Q}, that a lower limit exists for the mobility parameters
below which the map is not guaranteed to converge, and that the performance of the map
is progressively reduced when the model parameters approach such value.

For these reasons, in the present section we carry out
a numerical sensitivity analysis of $\overline{\mu}$ on the convergence of the fixed-point iteration.
As both hole and electron mobilities play a role in determining the
behaviour of the map and as the model we considered for them is parametrised
on the mobility values at zero electric field $\mu_{n,0}$ and $\mu_{p,0}$,
we decided to perform the following analyses:
\begin{itemize}
	\item decreasing $\mu_{p,0}$ and keeping $\mu_{n,0}$ fixed at the reference value;
	\item decreasing $\mu_{n,0}$ and keeping $\mu_{p,0}$ fixed at the reference value;
	\item setting $\mu_{p,0} = \mu_{n,0} = \tilde{\mu}$ and decreasing them.	
\end{itemize}
We expect the first two analyses to provide similar results
as the effect of the two charge carrier densities on the model is symmetrical whereas
in the third case we aim to assess whether the simultaneous change of
the two mobilities results in a combined effect.
Moreover, as previously done in Sect.~\ref{sec:num_changing_Q}, we consider
the same two operation modes characterised by
different values of the applied potential $\varphi_C - \varphi_A$
in order to determine whether this quantity has an impact on the convergence of the map
while changing the mobility parameter.

In Fig.~\ref{fig:iterations_mobility_bis}
we report the graphs of the number of iterations needed by the map to
converge as a function of the mobility parameter, for an applied voltage equal to
0.4\,V~(Fig.~\ref{fig:iterations_mobility_04V})
and 0\,V~(Fig.~\ref{fig:iterations_mobility_0V}) respectively.
We notice that in both cases the map performs almost similarly with the
reduction of either $\mu_{p,0}$ or $\mu_{n,0}$, requiring
a slightly larger number of iterations
for the $p$ case at 0.4\,V and for the $n$ case at 0\,V.
This asymmetry in the behaviour could be attributed to the marginally different
reference values for the hole and electron mobilities or to a difference
in the discretisation of domains $\Omega_n$ and $\Omega_p$ due to the
algorithm for the generation of the anisotropic meshes.
It is interesting though that when both mobilities are decreased simultaneously,
the convergence to the fixed point is slower and the threshold value is considerably
higher. This can be explained by the fact that by decreasing both mobilities,
the charge carrier density in the donor and acceptor materials increase and,
as highlighted in the discussion of Sect.~\ref{sec:num_changing_Q},
the second order term $2H\gamma n p$ at the interface
becomes more relevant, as both $n$ and $p$ increase.

\begin{figure}[h!]
	\centering
	\subfigure[]{
		\includegraphics[width=0.45\textwidth]{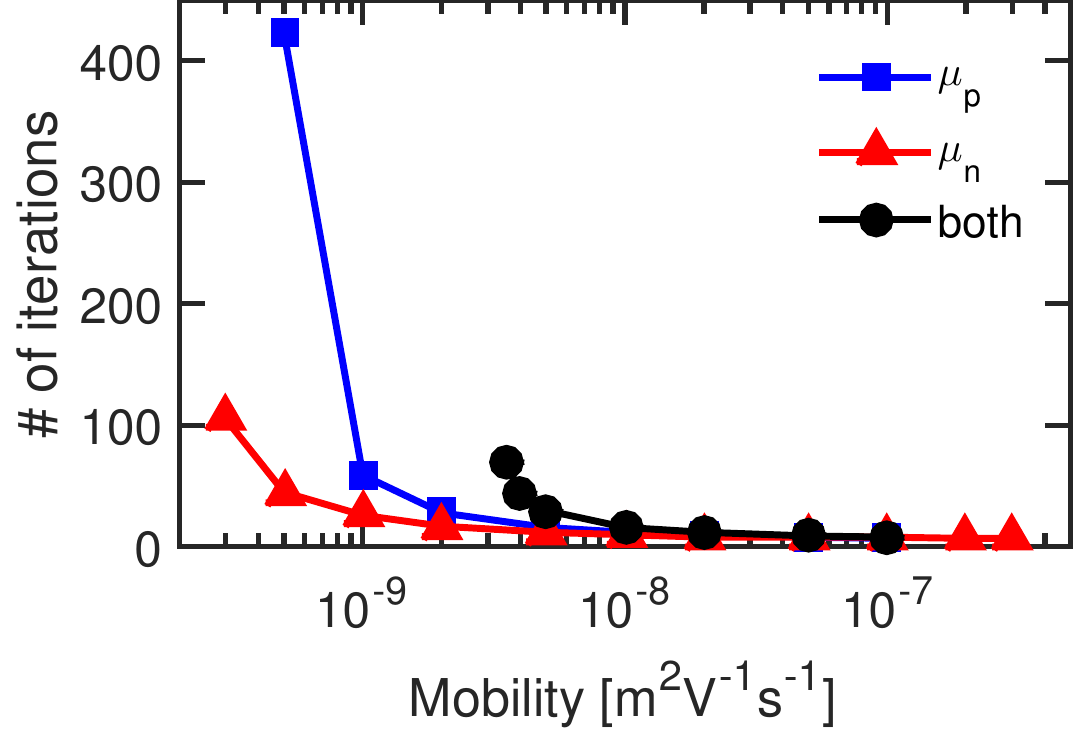}
		\label{fig:iterations_mobility_04V}
	}
	\subfigure[]{
		\includegraphics[width=0.45\textwidth]{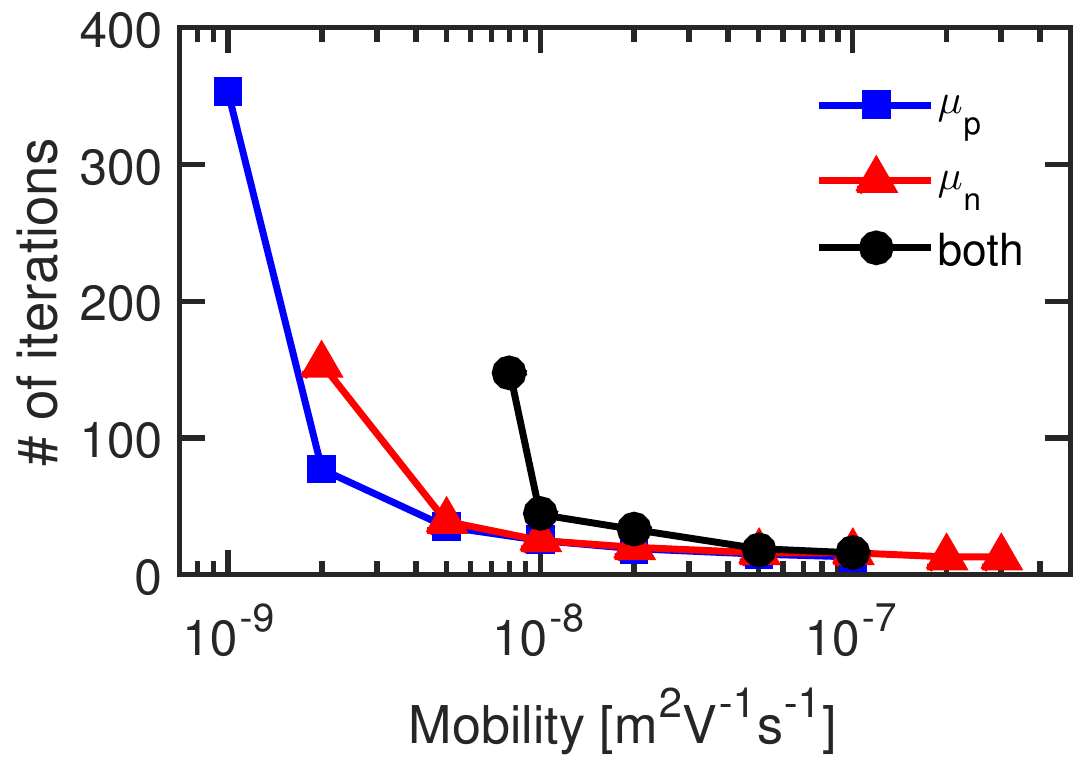}
		\label{fig:iterations_mobility_0V}
	}
	\caption{Number of iterations needed by the map to converge
		changing the value of the hole and electron mobility or
		both simultaneously with 0.4\,V (a) or 0\,V (b).}
	\label{fig:iterations_mobility_bis}
\end{figure}

\subsection{Changing the applied electric potential}
\label{sec:num_changing_applied_potential}
Finally we want to analyse the impact of the difference in the electric
potential between the electrodes $\varphi_{C} - \varphi_{A}$ on the convergence
properties of the map. The obtained analytical result states that convergence
to a unique fixed point is guaranteed if the applied voltage is
\emph{small enough}, similarly to what happens in semiconductor device modelling
using the Drift-Diffusion model~\cite{markowich1986stationary,Jerome:AnalyCharTran},
and we want to test if the map has a similar behaviour
as that observed when changing $Q$ and $\mu$.

Fig.~\ref{fig:convergence_potential} shows the number of iterations
needed by the map to satisfy~\eqref{eq:stopping_criterion} in a range of
applied voltages between $-1.5$\,V and $1.5$\,V.
Interestingly, convergence is observed for all the considered values and by
increasing the applied voltage (both for negative and positive values)
the convergence speed is enhanced.

\begin{figure}[h!]
	\centering
	\includegraphics[width=0.65\textwidth]{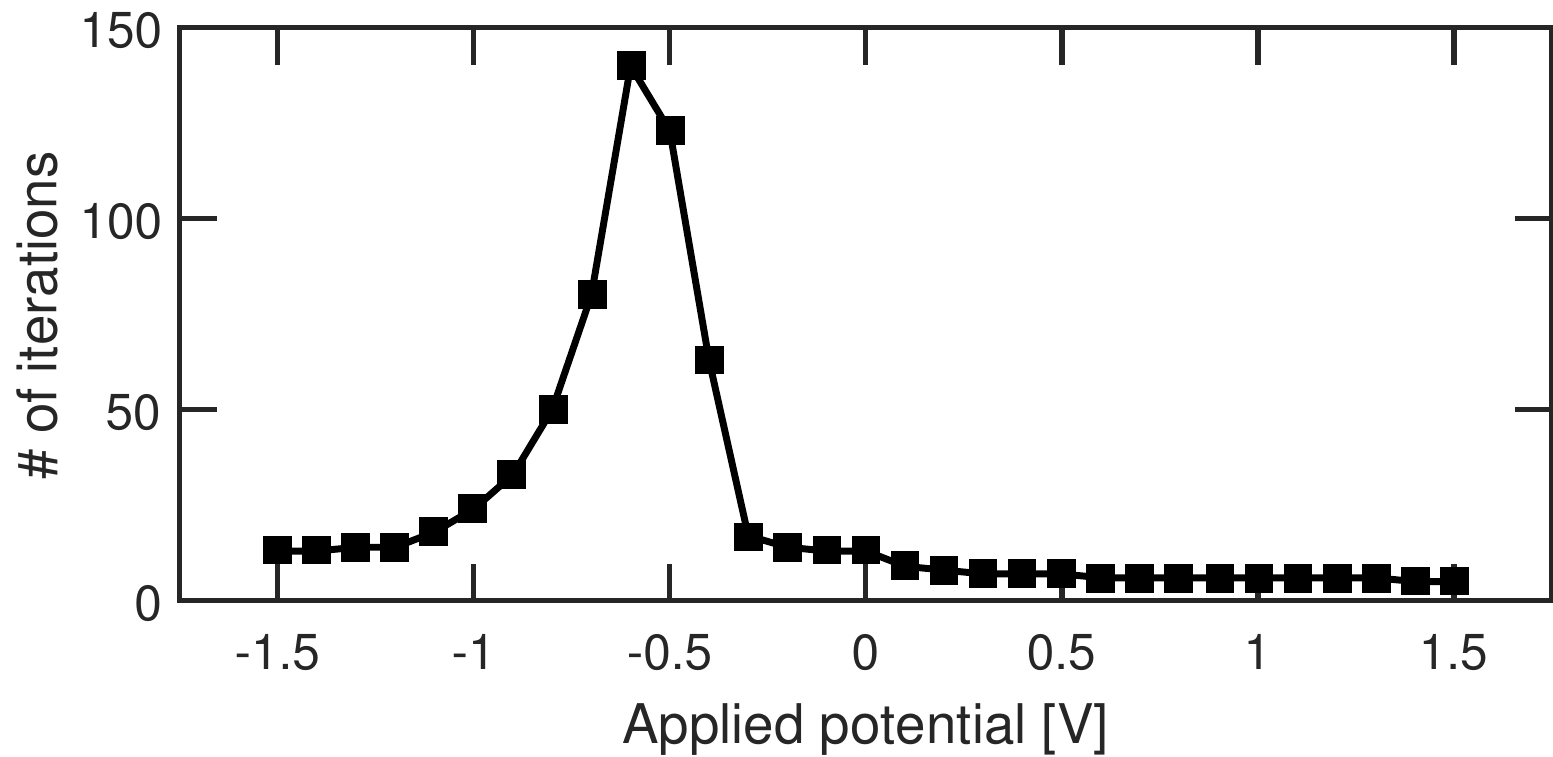}
	\caption{Number of iterations needed by the map  to converge
		changing the value of the potential difference
		$\varphi_C - \varphi_A$ at the electrodes.}
	\label{fig:convergence_potential}
\end{figure}

We already analyzed the operation mode with applied voltage equal to $0.4$\,V,
which can be assimilated to all the other configurations between $0$\,V and $1.5$\,V.
The generated electric field helps charges migrate
towards the electrodes, generating electric current.
The charge densities in device are hence relatively small and the nonlinear terms
in equations~\eqref{eq:excitons}-\eqref{eq:potential} are not dominant,
so the map does not need to perform many steps to meet the tolerance.

Convergence has been proven more difficult in the range of values
between $-0.8$\,V and $-0.4$\,V, with a distinct spike at $-0.6$\,V.
In this regime the applied voltage counteracts the potential difference
determined by the displacement of the dissociated charges generating an electric field
that tends to drive these latter back to the interface where they can recombine.
The main consequence is that the generated output current is close to zero
(\emph{open circuit} regime) and also the charge carrier
densities in the device are significantly larger than in the
current extracting operation mode, making the nonlinear
terms more important and reducing the convergence speed of the iteration map.

Furtherly decreasing the applied potential below $-0.8$\,V,
we observe again an improvement in the performance of the map.
In these configurations, the applied electric field is strong
enough to move most of the generated charge carriers back to the interface
where they recombine, reducing considerably the carrier densities and hence the
nonlinear effects.

\subsection{Further testing of map convergence: the use of Einstein's relation}
The aim of this section is to analyse the behaviour of the iterative map
in configurations where the model parameter definitions do not satisfy the
assumptions of Tab.~\ref{tab:assumptions} made in
order to prove the results of Sect.~\ref{sec:fixed_point_map}.
A significant case is that obtained by considering the mobility parameter
definition as in~\eqref{eq:mob_model} but with no ceiling for high electric
field values
\begin{equation}
\mu _{i}\left( E\right) =\mu _{i,0}\exp \left( \dfrac{\beta _{i}\sqrt{E}}{%
k_{B}T}\right)   \label{eq:mobility_model_no_ceiling}
\end{equation}%
and assuming the Einstein-Smoluchowski relation to hold~\cite%
{Jerome:AnalyCharTran}, i.e.
\begin{equation}
D_{i}^{\mathrm{einstein}}\left( E\right) =\dfrac{k_{B}T}{q}\,\mu _{i}\left(
E\right) .  \label{eq:einstein_model}
\end{equation}%
The study of this configuration is of particular interest as it is
frequently used in the literature on the topic~\cite%
{porro2014,deFalco2012,williams2008,walker2008}.

Upon changing the values for the generation term $Q$ (see Fig.~\ref%
{fig:convergence_generation_einstein}) and the zero-field mobility (see Fig.~%
\ref{fig:iterations_mobility_bis_einstein}), the map shows a performance
similar to that observed in the previous analyses. In particular, comparing
with results presented in Sections~\ref{sec:num_changing_Q} and~\ref%
{sec:num_changing_mu} the number of iterations needed by the map to converge
is generally higher in the configurations with no applied potential
(0\thinspace V). In such configuration, the electric field in the device is
close to zero, so that the diffusion coefficients are smaller than predicted
by~\eqref{eq:diffusion_parameters}, that is
\begin{equation}
D_{i}^{\mathrm{einstein}}(E\simeq 0)\simeq \dfrac{k_{B}T}{q}\mu _{i}(0)<\dfrac{%
k_{B}T}{q}\ \dfrac{\mu _{i}(0)+\mu _{i}(E^{\ast })}{2}=D_{i}^{\mathrm{const}}
\label{eq:estimate_D}
\end{equation}%
as~\eqref{eq:mob_model} is a monotonically increasing function of $E$.

\begin{figure}[h!]
	\centering
	\includegraphics[width=0.5\textwidth]{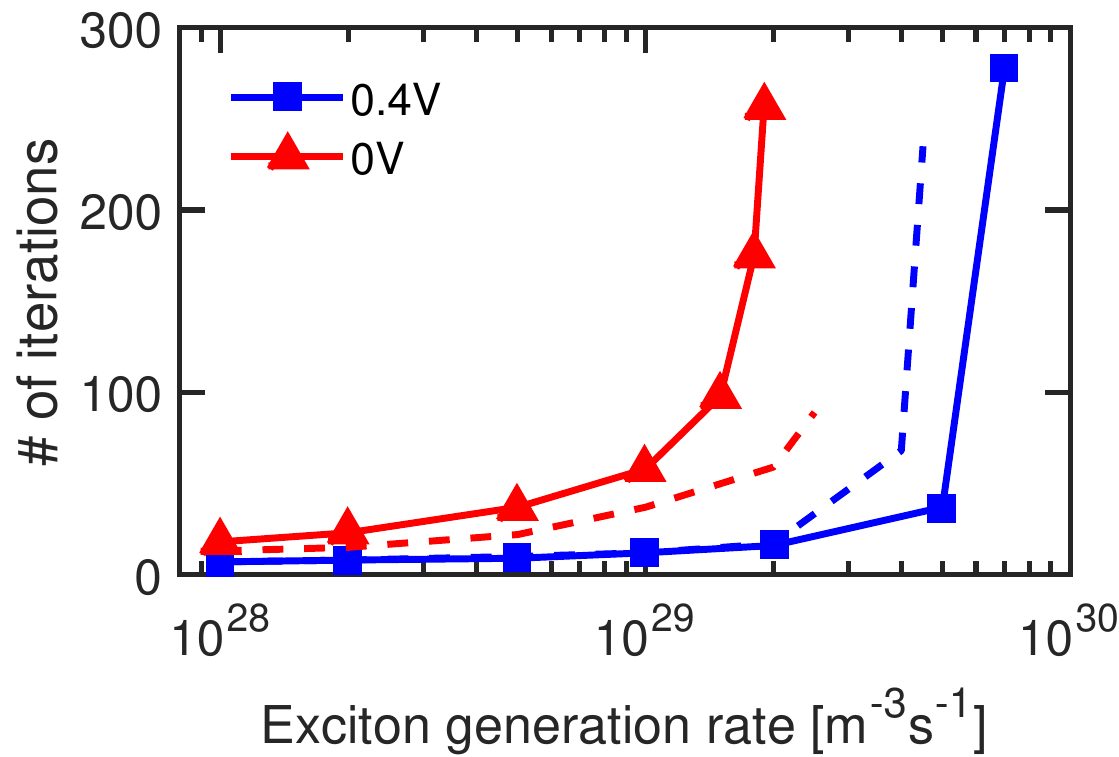}
	\caption{Number of iterations needed by the map to converge
		changing the value of the exciton generation rate.
        Results of Section~\ref{sec:num_changing_Q} are
        displayed with dotted lines.}
	\label{fig:convergence_generation_einstein}
\end{figure}

\begin{figure}[h!]
	\centering
	\subfigure[]{
		\includegraphics[width=0.45\textwidth]{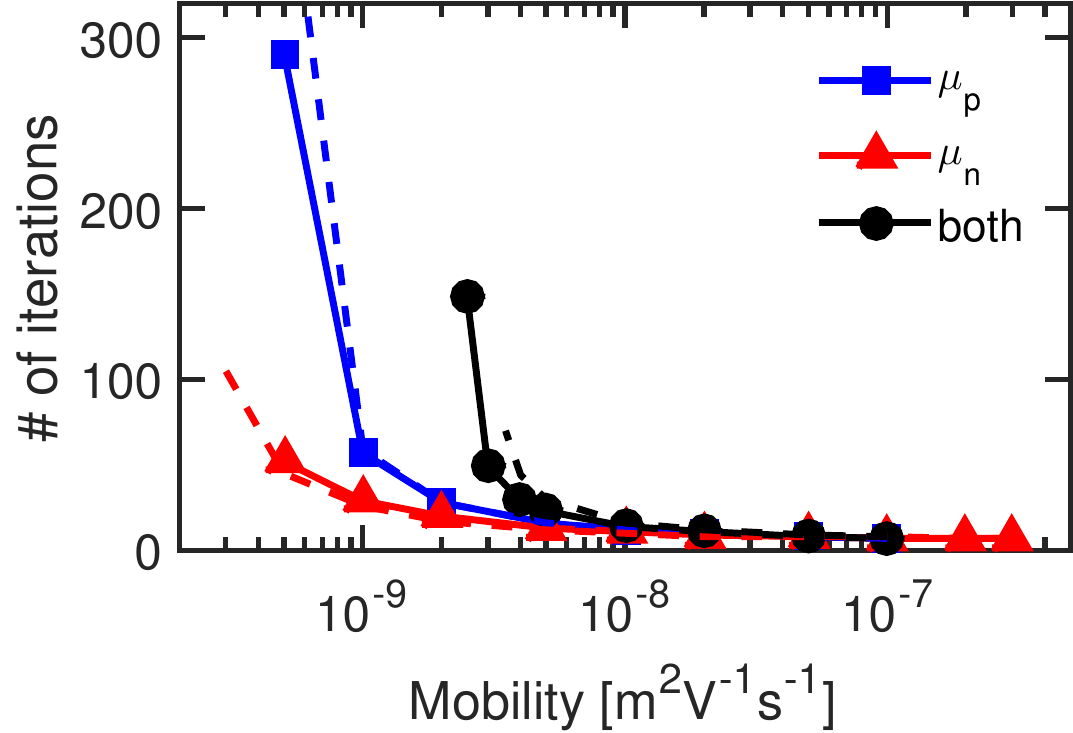}
		\label{fig:iterations_mobility_04V_einstein}
	}
	\subfigure[]{
		\includegraphics[width=0.45\textwidth]{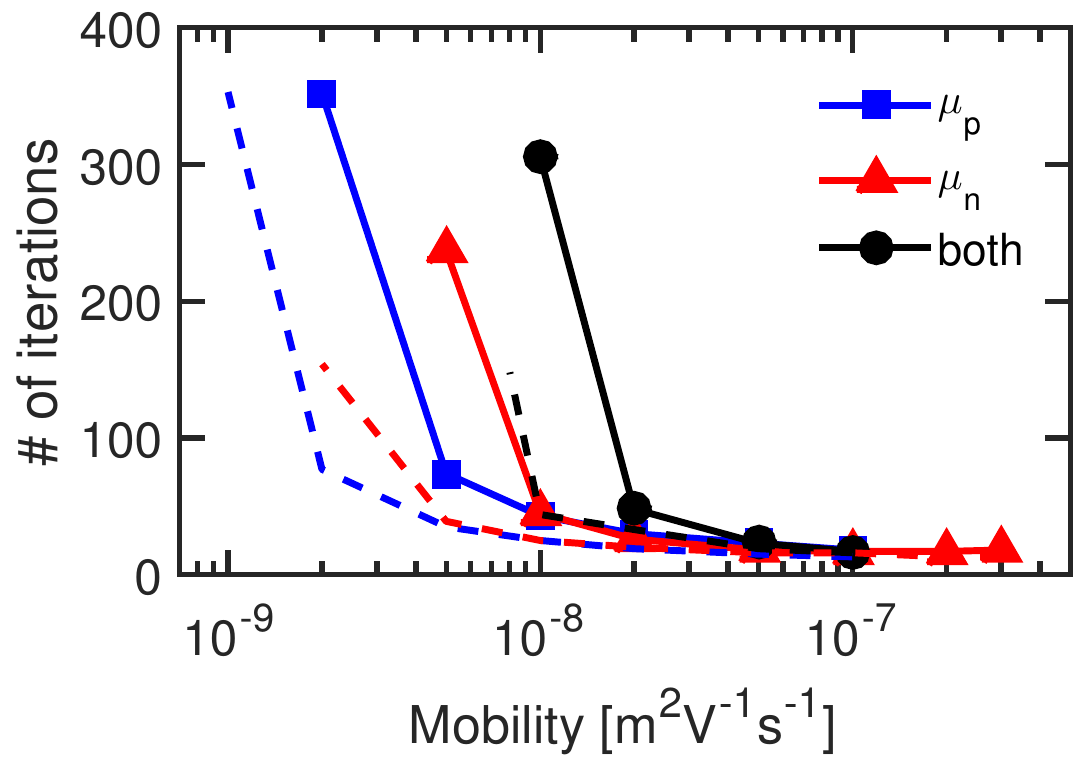}
		\label{fig:iterations_mobility_0V_einstein}
	}
	\caption{Number of iterations needed by the map to converge
		changing the value of the hole and electron mobility or
		both simultaneously with 0.4\,V (a) or 0\,V (b).
        Results of Section~\ref{sec:num_changing_mu} are
        displayed with dotted lines.}
	\label{fig:iterations_mobility_bis_einstein}
\end{figure}

Upon changing the value of the applied potential, we can observe an
interesting behaviour of the iterative map,
see Fig.~\ref{fig:convergence_potential_einstein}.
\begin{figure}[h!]
	\centering
	\includegraphics[width=0.65\textwidth]{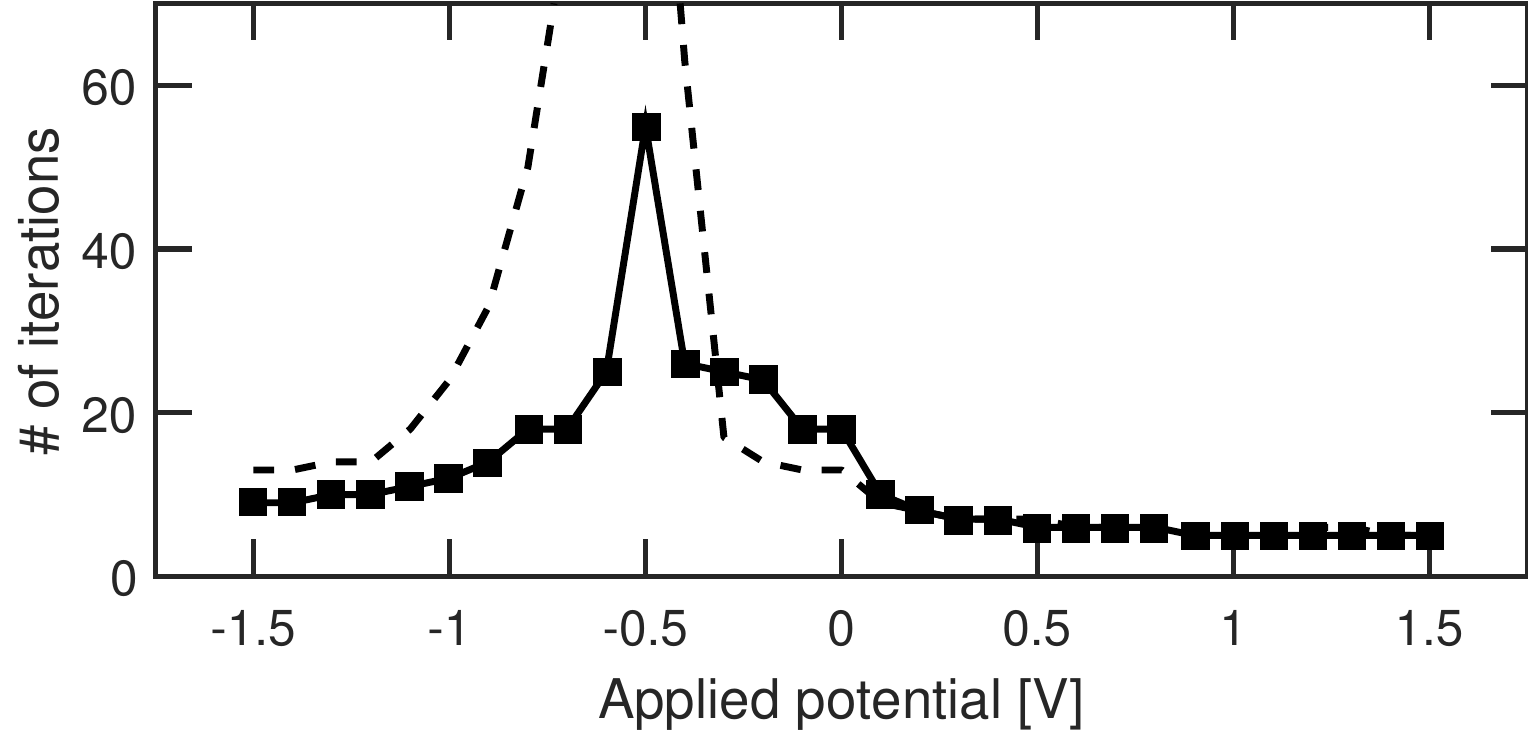}
	\caption{Number of iterations needed by the map  to converge
		changing the value of the potential difference
		$\varphi_C - \varphi_A$ at the electrodes.
        Results of Section~\ref{sec:num_changing_applied_potential} are
        displayed with the dotted line.}
	\label{fig:convergence_potential_einstein}
\end{figure}
The number of iterations needed for the map to converge are the same
as reported in Sect.~\ref{sec:num_changing_applied_potential} for
values of applied voltage strictly greater than 0\,V. This is to be ascribed to the
fact that the electric field in the device is large enough to make the drift term in
the current density for electrons and holes dominant with respect to the diffusive term,
in such a way that the different mathematical representations of the diffusion coefficient
play no role in this branch of values of the applied external electric force.
Things change in the range of values of applied potential between 0\,V and -0.3\,V.
Indeed, in this working regime the electric field is close to zero, or small, so that
relation~\eqref{eq:estimate_D} predicts
$D_{i}^{\mathrm{einstein}}<D_{i}^{\mathrm{const}}$ and thus,
consistently, photogenerated electrons and holes hardly diffuse from the interface region
making the effect of nonlinear bimolecular recombination terms more relevant and
requiring a (slightly) higher number of iterations for the map to converge.
However, furtherly increasing the (negative) value of the applied voltage
produces an increase of the strength of the electric field in the device, so that
relation~\eqref{eq:estimate_D} does no longer hold and we have that
$D_{i}^{\mathrm{einstein}}>D_{i}^{\mathrm{const}}$.
As a consequence, the diffusive term in the current density helps photogenerated charges
to detach from the interface region and move towards the contacts, this having
the effect to reduce considerably the number of iterations for the map to converge
as illustrated by the dashed line in Fig.~\ref{fig:convergence_potential_einstein}.
The sharp peak at -0.5\,V corresponds to the open circuit conditions, and this explains
the sharp increase of number of iterations: drift and diffusive current densities mutually
cancel so that charges are confined at the interface and nonlinear recombination makes
the convergence of the map to slow down. Then, for larger negative values of the applied
voltage the behaviour of the device is again dominated by drift current densities, so that
the convergence of the functional iteration becomes insensitive to the adopted model
of the diffusion coefficient.

\section{Conclusions and perspectives}\label{sec:conclusions}
In this article we have addressed the analytical study of a multidomain system of nonlinearly
coupled PDEs, with nonlinear transmission conditions at the material interface, that represents the
mathematical modeling picture of an organic solar cell. The system is constituted by
a set of conservation laws for four distinct species: excitons and polarons (electrically neutral),
and electrons and holes (negatively and positively charged). The analysis is conducted
in the stationary regime and under assumptions on the parameters and data that make
the considered problem a close representation of a realistic nanoscale device for energy photoconversion.
The resulting problem is a highly nonlinearly coupled system of advection-diffusion-reaction PDEs
for which existence and uniqueness of weak solutions, as well as nonnegativity of concentrations, is proved
via a solution map that is a variant of the Gummel iteration commonly used in the
treatment of the DD model for inorganic semiconductors. Results are established upon
assuming suitable restrictions on the data and some regularity property on the mixed boundary value problem for the Poisson equation. The main analytical conclusions are numerically validated through an extensive
sensitivity analysis devoted to characterizing the dependence of the convergence of the fixed-point iteration
on the most relevant physical parameters of the OSC. Simulation predictions are in excellent agreement
with theoretical limitations and suggest that failure to convergence may principally occur in the following
three distinct conditions:
\begin{itemize}
\item when the exciton generation rate $Q$ becomes too large;
\item when carrier mobility becomes too small;
\item when the device works close to open-circuit conditions.
\end{itemize}
We believe that such conclusions may provide useful indications to improve, on the one hand, the
development of efficient solution algorithms to be implemented in computational tools, and, on
the other hand, the search of suitable materials in view of an optimal design of a organic solar cell
of the next generation. We also believe that the functional techniques employed in the present article
may be profitably adopted in the analysis of well-posedness of the multidomain nonlinear model
in the time-dependent case. This aspect will be the object of our next investigation.

\bibliographystyle{plain}      
\bibliography{biblio}   

\end{document}